\documentclass[10pt]{article}
\usepackage[top=0.9in,bottom=1.1in,left=1.05in,right=1.05in]{geometry}
\usepackage{amsmath,amssymb}
\usepackage{amsthm}
\usepackage{threeparttable}
\usepackage{latexsym}
\usepackage{mathrsfs,dsfont}
\usepackage{cancel}
\usepackage{comment}
\usepackage{float}
\usepackage{graphicx}
\usepackage{epstopdf}
\usepackage{color}
\usepackage{enumerate}
\usepackage{mathtools}
\mathtoolsset{showonlyrefs}
\DeclareGraphicsExtensions{.eps}
\numberwithin{equation}{section}

\newcommand{\MCR}{\mathcal{R}}

\newcommand{\MCA}{\mathcal{A}}
\newcommand{\MCM}{\mathcal{M}}
\newcommand{\MCF}{\mathcal{F}}
\newcommand{\MCO}{\mathcal{O}}
\newcommand{\MCD}{\mathcal{D}}

\newcommand{\MCL}{\mathcal{L}}

\newcommand{\EE}{\mathbb{E}}

\newcommand{\RR}{\mathbb{R}}
\newcommand{\NN}{\mathbb{N}}

\newcommand{\Ltx}{\mathcal{L}_{t,x}}

\newcommand{\Vf}{V^{\eps,\delta}}

\newcommand{\vz}{v^{(0)}}

\newcommand{\pz}{{\pi^{(0)}}}

\newcommand{\vzl}{v^{\pz,(0)}}

\newcommand{\vtl}{v^{\pz,(2)}}

\newcommand{\Vfl}{V^{\pz,\eps,\delta}}
\newcommand{\Vft}{\widetilde{V}^{\eps,\delta}}
\newcommand{\pozt}{\widetilde \pi^{(1,0)}}
\newcommand{\pzot}{\widetilde{\pi}^{(0,1)}}
\newcommand{\Vflz}{V^{\pz,\eps,0}}
\newcommand{\Vflo}{V^{\pz,\eps,1}}
\newcommand{\vzol}{v^{\pz,(0,1)}}
\newcommand{\vozl}{v^{\pz,(1,0)}}
\newcommand{\vthzl}{v^{\pz,(3,0)}}

\newcommand{\vtol}{v^{\pz, (2,1)}}
\newcommand{\vool}{v^{\pz, (1,1)}}
\newcommand{\vtzl}{v^{\pz, (2,0)}}
\newcommand{\vzo}{v^{(0,1)}}
\newcommand{\voz}{v^{(1,0)}}
\newcommand{\vztw}{v^{(0,2)}}
\newcommand{\voo}{v^{(1,1)}}
\newcommand{\vtz}{v^{(2,0)}}
\newcommand{\vzot}{\widetilde  v^{(0,1)}}
\newcommand{\vozt}{\widetilde  v^{(1,0)}}

\newcommand{\pitilde}{\widetilde \pi^{\eps, \delta}}

\newcommand{\pzt}{\widetilde\pi^0}

\newcommand{\vzt}{\widetilde{v}^{(0)}}

\newcommand{\aves}{\overline\lambda}

\newcommand{\eps}{\epsilon}

\newcommand{\abs}[1]{\left|#1\right|}
\newcommand{\average}[1]{\left\langle#1\right\rangle}

\newcommand{\ud}{\,\mathrm{d}}

\newcommand{\half}{\frac{1}{2}}

\newtheorem{theo}{Theorem}[section]
\newtheorem{lem}[theo]{Lemma}

\newtheorem{prop}[theo]{Proposition}
\newtheorem{assump}[theo]{Assumption}

\begin{document}

\title{\vspace{-50pt} Multiscale Asymptotic Analysis for Portfolio Optimization under Stochastic Environment}
%\author{Ruimeng Hu\\Department of Statistics and Applied Probability \\ University of California, Santa Barbara}
\author{Jean-Pierre Fouque\thanks{Department of Statistics \& Applied Probability,
 University of California,
        Santa Barbara, CA 93106-3110, {\em fouque@pstat.ucsb.edu}. Work  supported by NSF grant  DMS-1814091.}
        \and Ruimeng Hu\thanks{Department of Statistics, Columbia University, New York, NY 10027-4690, {\em rh2937@columbia.edu}.}
        }
\date{\today}
\maketitle

\begin{abstract}
	
Empirical studies indicate the presence of multi-scales in the volatility of underlying
assets: a fast-scale on the order of days and a slow-scale on the order of months. In our previous works, we have studied the portfolio optimization problem in a Markovian setting under each single scale, the slow one in  [Fouque and Hu, SIAM J. Control Optim., 55 (2017), 1990-2023], and the fast one in [Hu, Proceedings of the 2018 IEEE CDC, 5771-5776, 2018]. %That is, the volatility of the underlying asset is driven by a exogenous process, parameterized by $\eps$ or $\delta$ to capture the fast mean-reverting or slowly varying property. 
This paper is dedicated to the analysis when the two scales coexist in a Markovian setting. We study the terminal wealth utility maximization problem when the volatility is driven by both fast- and slow-scale factors. We first propose a zeroth-order strategy, and rigorously establish the first order approximation of the associated problem value. This is done by analyzing the corresponding  linear partial differential equation (PDE) via  regular and singular perturbation techniques, as in the single-scale cases. Then, we show the asymptotic optimality of our proposed strategy by comparing its performance  to admissible strategies of a specific form. Interestingly, we highlight that a pure PDE approach does not work in the multi-scale case and, instead,  we use  the so-called epsilon-martingale decomposition. This completes the analysis of portfolio optimization in both fast mean-reverting and slowly-varying Markovian stochastic environments.

\end{abstract}


%\begin{AMS}
%	93E20, 91G10, 35Q93, 35C20
%\end{AMS}

\section{Introduction}\label{sec_intro}
A classical problem in mathematical finance is the utility maximization of consumption and/or terminal wealth for an investor. This was first studied by Mossin \cite{Mo:68} and Samuelson \cite{Sa:69} in the discrete-time framework, and by Merton \cite{Me:69, Me:71} in the continuous-time case. In Merton's seminal work, the underlying assets follow the Black-Scholes (BS) model, that is, the return and volatility are constants, and the utility is of certain type, for instance, the Constant Relative Risk Aversion utility function. Explicit solutions are provided on how to invest and/or consume. Later, this problem has been studied extensively in various settings and levels of generality, for example, to allow transaction costs \cite{MaCo:76, GuMu:13}, drawdown constraints \cite{GrZh:93, CvKa:95, ElTo:08},  to consider price impact \cite{CuCv:98}, and to extend the BS model to stochastic volatility \cite{Za:99, ChVi:05, FoSiZa:13, LoSi:16}.

This paper generalizes Merton's work in two directions. Firstly, observing time-varying risk aversion in individual asset allocation \cite{BrNa:08}, we consider \emph{general utility functions}. Secondly, in the direct of asset modeling, we use a multiscale stochastic volatility model, in line with the existence a fast-time scale in stock price volatility on the order of days as well as a slow-scale on the order of months in the financial markets \cite{FoPaSiSo:11}. 
In this context, asymptotic analysis has been developed over decades to option pricing problems, where singular and regular perturbation methods are applied to derive efficient approximations; here, we present new results for the nonlinear Merton problem with general utility functions on $\RR^+$.

 Following the modeling  in \cite{FoPaSiSo:11}, the Markovian dynamics of the asset price and stochastic factors read:
	\begin{align}
	&\ud S_t  = \mu(Y_t,Z_t)S_t \ud t + \sigma(Y_t,Z_t)S_t\ud W_t,\\
	&\ud Y_t = \frac{1}{\eps}b(Y_t)\ud t + \frac{1}{\sqrt{\eps}}a(Y_t)\ud W_t^Y,\\
	&\ud Z_t = \delta c(Z_t)\ud t + \sqrt\delta g(Z_t)\ud W_t^Z,
	\end{align}
	where the asset return $\mu$ and volatility $\sigma$ are functions of the two factors: the fast-scale factor  $Y_t$ and the slow-scale one $Z_t$. The two parameters  $(\eps,\delta) \ll 1$ are small to capture the two scales, $(W_t, W^Y_t, W_t^Z)$ are correlated Brownian motions. Detailed discussion about this modeling is presented in Section~\ref{sec_PDE_multi}. 
	
	Fixing a time horizon $T$, we are interested in the terminal utility maximization problem:
	\begin{equation}\label{def_object}
	\sup_\pi \EE[U(X_T^\pi)],
	\end{equation}
	where the general utility $U(\cdot)$ satisfies Assumption~\ref{assump_U}, and $X_t^\pi$ is the investor's wealth at time $t$, consisting of two parts: the money invested in the risky asset $S_t$, denoted by $\pi$, and the remaining part $X_t^\pi - \pi$ put into the money market earning the risk-free interest rate. Restricting $\pi$ to self-financing strategies (no exogenous deposit or withdrawal of money after time 0), and assume $r = 0$ for simplicity, $X_t^\pi$ follows:
	\begin{equation}\label{def_Xt}
	\ud X_t^\pi = \pi_t\mu(Y_t,Z_t)\ud t + \pi_t\sigma(Y_t,Z_t) \ud W_t.
	\end{equation}
	
	Focusing on the feedback strategies, that is, let $\pi_t = \pi(t, X_t^\pi, Y_t, Z_t)$, this problem can be tackled using the method of dynamic programming. The main idea is to embed our original problem \eqref{def_object} into a much larger class of problems with different starting time $t$, initial wealth $x$ and initial levels of both factors $(y,z)$, and then to connect all these problems together with a PDE known as the Hamilton-Jacobi-Bellman (HJB) equation. For that purpose, we define the value function as
	\begin{equation}\label{def_Vf}
	\Vf(t,x,y,z)  = \sup_{\pi \in \MCA^{\eps, \delta}} \EE\left[U(X_T^\pi)\vert X_t^\pi = x, Y_t = y, Z_t = z\right],
	\end{equation}
	with $\MCA^{\eps, \delta}$ collecting all feedback admissible controls:
	\begin{equation}\label{def_MCA}
	\MCA^{\eps, \delta} = \{\pi: \pi_t \equiv \pi(t, X_t^\pi, Y_t, Z_t), X_s^\pi \geq 0, \;\forall s \geq t, \text{ given } (X_t^\pi, Y_t, Z_t) = (x, y, z)\}.
	\end{equation}
	The superscripts $\eps, \delta$ emphasize the dependence on the two small parameters introduced through $Y_t$ and $Z_t$.

	In general, depending on assumptions, $\Vf$ is characterized as a classical or viscosity solution to the HJB equation \eqref{p2_eq_VfPDE}, for which closed-form solutions are rarely available. In \cite{FoSiZa:13}, assuming the existence of classical solutions, a formal first order expansion is derived via singular and regular perturbation techniques: 
	\begin{equation}\label{eq_Vfexpansion}
	\Vf = \vz + \sqrt{\eps}\voz + \sqrt{\delta}\vzo + \eps\vtz + \delta \vztw + \sqrt{\eps\delta}\voo + \cdots.
	\end{equation}
	Formulations of  $\vz$, $\voz$ and $\vzo$ will be presented in Section~\ref{p2_sec_multiheuristic}. Note that the above expansion is not rigorous even in the canonical power utility case, as the distortion transformation \cite{Za:99} which linearizes the problem is not available with more than one stochastic volatility factor. Nevertheless, they conjecture \cite[Section~4.2]{FoSiZa:13} that a zeroth order strategy, defined based on the leading order term $\vz$ in \eqref{eq_Vfexpansion}:
	\begin{equation}
	\pz = -\frac{\lambda(y,z)}{\sigma(y,z)}\frac{\vz_x(t,x,z)}{\vz_{xx}(t,x,z)},
	\end{equation}
	can reproduce $\Vf$ up to the first order correction, that is, the value function associated to $\pz$ takes the form 
	 $\vz+ \sqrt{\eps}\voz + \sqrt{\delta}\vzo + o(\sqrt{\eps} + \sqrt{\delta})$.

	 \medskip
	 \noindent{\bf Main results.} The goal of this paper is twofold: Firstly,  we rigorize the above assertion of $\pz$. To this end, we analyze the linear PDE satisfied by the problem value associated to $\pz$:
	 \begin{equation}
	 \Vfl := \EE\left\{U(X_T^\pz)\Big\vert X_t^\pz = x, Y_t = y, Z_t = z\right\},
	 \end{equation}
	 where $X_t^\pz$ is given in \eqref{def_Xt} with $\pi = \pz$. A rigorous first order approximation is obtained for $\Vfl$, which coincides with $\vz + \sqrt{\eps}\voz  + \sqrt{\delta}\vzo$. This leads to our first result.
	 	\begin{theo}\label{p2_Thm_full}
	 		Let $\vz$, $\voz$ and $\vzo$ be the coefficient functions from the heuristic expansion of $\Vf$ in \eqref{eq_Vfexpansion},
	 		identified in Section~\ref{p2_sec_multiheuristic}. Under Assumptions~\ref{assump_U} and \ref{p2_assump_valuefuncadd}
	 		the residual function $E(t,x,y,z)$ defined by
	 		$$E(t,x,y,z) := \Vfl(t,x,y,z) - \vz(t,x,z) - \sqrt{\eps}\voz(t,x,z) -\sqrt{\delta}\vzo(t,x,z), $$
	 		is of order $\eps + \delta$, for all $(t,x,y,z) \in [0,T]\times \RR^+\times \RR\times \RR$. That is, $\abs{E(t,x,y,z)} \leq C(\eps + \delta)$, where $C$ may depend on $(t,x,y,z)$ but not on $(\eps, \delta)$.
	 	\end{theo}
	 	The proof will be given in Section~\ref{p2_sec_multipz}.

	 Secondly, and more importantly, we show that $\pz$ outperforms any admissible strategy of a certain form. To be precise, we compare its performance to the one of 
	 \begin{equation}\label{def_pitilde}
	  \widetilde\pi^{\eps, \delta}= \pzt + \eps^\alpha \pozt + \delta^\beta\pzot,
	  \end{equation} 
	 for some processes $\pzt_t, \pozt_t, \pzot_t$ (not necessarily in the feedback form) and some positive powers $\alpha, \beta$. To this end, for fix choices of $\pzt$, $\pozt$ and $\pzot$ of which the assumptions are postponed to Section~\ref{p6_sec_optimality} and Appendix~\ref{p6_appendix_addasump},
	we denote by $\Vft$ the value process:
	 \begin{equation}\label{def_Vft}
	 \Vft_t = \EE\{U(X_T^\pi)\vert \MCF_t\}, \quad \pi = \widetilde \pi^{\eps, \delta} \text{ in } \eqref{def_pitilde},
	 \end{equation}
	 with $\MCF_t = \sigma(W_s, W_s^Y, W_s^Z, s \leq t)$.
	The asymptotic optimality of $\pz$ is then obtained by comparing the approximation of $\Vfl$ to the one of $\Vft$. We now summarize this result as follows.
	\begin{theo}\label{p6_thm_main}
			Under Assumptions~\ref{assump_U}, \ref{p2_assump_valuefuncadd},  \ref{p6_assump_piregularity} and \ref{p6_assump_optimality}, for any fixed choice of $\pzt, \pozt, \pzot, \alpha, \beta$,   the following limit exists in $L^1$ and satisfies
			 %family of trading strategies \\
			 %$\widetilde \MCA^{\eps,\delta}[\pzt, \pozt,\pzot, \alpha, \beta]$, 	
			\begin{equation}\label{def_inq}
			\ell := \lim_{(\eps, \delta) \to 0}\frac{\Vft_t - \Vfl(t,X_t, Y_t, Z_t)}{\sqrt{\eps} + \sqrt\delta} \leq 0, \text{ in } L^1.
			\end{equation}
			That is, the strategy $\pz$ that generates $\Vfl$ performs asymptotically better up to order $\sqrt\eps + \sqrt{\delta}$ than the strategy $\pitilde$ given in \eqref{def_pitilde}.
			%		Moreover, the inequality can be written according to the following four cases:
			%		\begin{enumerate}[(i)]
			%			\item $\pzt = \pz$, $\alpha > H/2$: $\ell = 0$ and $\Vzp_t = \Vzl_t + o(\delta^H)$;
			%			\item $\pzt = \pz$, $\alpha = H/2$: $-\infty < \ell < 0$ and $\Vzp_t= \Vzl_t + \MCO(\delta^H)$ with $\MCO(\delta^H) <0$;
			%			\item $\pzt = \pz$, $\alpha < H/2$: $\ell = -\infty$ and $\Vzp_t = \Vzl_t + \MCO(\delta^{2\alpha})$ with $\MCO(\delta^{2\alpha}) <0$;
			%			\item $\pzt \neq \pz$: $\lim_{\delta \to 0} \Vzp_t < \lim_{\delta \to 0} \Vzl_t$,
			%		\end{enumerate}
			%		where all relations between $\Vzp_t$ and $\Vzl_t$ hold under $L^2$ sense.
		\end{theo}
		
The reason to consider such a form of $\widetilde \pi^{\eps, \delta}$ is the following. Under mild assumptions, the optimizer to problem \eqref{def_Vf}, denoted by $\pi^\ast$, exists \cite{KrSc:03}. Although $\pi^\ast$ has clear dependence on  $(\eps, \delta)$, it is unknown whether it will converge as $(\eps, \delta)$ go to zero. Supposing that  $\pi^\ast$ admits a limit, say $\pzt$, then it is natural to consider  $\widetilde \pi^{\eps, \delta}$ as a first order perturbation of the limiting $\pzt$. The  parameters $(\alpha, \beta)$ allow for corrections of any positive powers of $(\eps, \delta)$, giving more flexibility to  this perturbation.

Several remarks regarding to our results: firstly, our model considers two volatility factors, one fast and one slow, simultaneously. This extends our previous work \cite{FoHu:16} and \cite{Hu:18}, where the return $\mu$ and volatility $\sigma$ are driven by a single factor. In turn, it requires to combine together the regular and singular perturbation techniques, which are applied separately in aforementioned work. This involves nontrivial additional difficulties.
Secondly, we work with general utility functions, as oppose to a certain type of utility (power, exponential, log, etc.) considered by the majority of literature. This generalization is important since not everyone's utility is of CRRA type \cite{BrNa:08}. Thirdly, although we are not able to fully characterize $\Vf$ by justifying the expansion \eqref{p2_eq_Vfexpansion}, we partially answer this question by analyzing a suboptimal strategy $\pz$ which has the rigorous first order approximation coincide with the heuristics of $\Vf$ and outperforms any admissible stratetgy of the form \eqref{def_pitilde}. 

\medskip
\noindent{\bf Organization of the paper.} In Section~\ref{sec_PDE_multi}, we restate the multi-scale model and the heuristic expansion results in \cite[Section 4]{FoSiZa:13}. We then briefly review the classical Merton problem, which is closely related to the zeroth-order value $\vz$ in \eqref{p2_eq_Vfexpansion} and to the derivations in later sections. We also list all needed assumptions and lemmas as a preparation of later proofs. The performance of $\pz$-portfolio is analyzed and its first order approximation is rigorously derived in Section~\ref{p2_sec_multipz}. Section~\ref{p6_sec_optimality} is dedicated to the asymptotic optimality of $\pz$, as phrased in Theorem~\ref{p6_thm_main}, by comparing the performance of $\pitilde =\pzt + \eps^\alpha\pozt + \delta^\beta\pzot$ with $\pz$ up to the first order. We make conclusive remarks in Section~\ref{sec_conclusion}.

\section{Preliminaries and assumptions}\label{sec_PDE_multi}

In this section, we first detail the multiscale stochastic volatility modeling, review the Merton PDE and risk-tolerance function,  summarize the expansion results in \cite[Section~4]{FoSiZa:13}, and list model assumptions on the utility and state processes. 

Recall the stochastic environments driven by a fast factor $Y_t$ and a slow factor $Z_t$, the underlying asset follows
	\begin{align}\label{p2_eq_Stfull}
	&\ud S_t  = \mu(Y_t,Z_t)S_t \ud t + \sigma(Y_t,Z_t)S_t\ud W_t,\\
	&\ud Y_t = \frac{1}{\eps}b(Y_t)\ud t + \frac{1}{\sqrt{\eps}}a(Y_t)\ud W_t^Y,\label{p2_eq_Ytfull}\\
	&\ud Z_t = \delta c(Z_t)\ud t + \sqrt\delta g(Z_t)\ud W_t^Z, \label{p2_eq_Ztfull}
	\end{align}
	where the standard Brownian motions $(W_t,W_t^Y,W_t^Z)$ are correlated by:
	\begin{equation}
	\ud \average{W,W^Y}_t = \rho_1 \ud t, \quad \ud \average{W,W^Z}_t = \rho_2 \ud t, \quad \ud \average{W^Y,W^Z}_t = \rho_{12} \ud t,
	\end{equation}
	with positive definite constraints: $\abs{\rho_1} <1$, $\abs{\rho_2} <1$, $\abs{\rho_{12}} <1$ and $1+2\rho_1\rho_2\rho_{12}-\rho_1^2-\rho_2^2-\rho_{12}^2 > 0$. Assumptions on the coefficients $\mu(y,z)$, $\sigma(y,z)$, $b(y)$, $a(y)$, $c(z)$ and $g(z)$ of the model will be specified in Section~\ref{sec_modelassump}. Both $\eps$ and $\delta$ are small positive parameters that characterize the fast mean-reversion of $Y_t$ and  slow variation of $Z_t$, respectively. The time-changed process $Z_t \stackrel{\MCD}{=}Z_{\delta t}^{(1)}$ is continuous and possesses a $\delta$-free infinitesimal generator, denoted by $\MCM_2$:
	\begin{equation}\label{p1_eq_Mgenerator}
	\MCM_2 = \half g^2(z)\partial_{z}^2 + c(z)\partial_z.
	\end{equation}
 	Similarly, the process $Y_t^{(1)} \stackrel{\MCD}{=} Y_{\eps t}$ has the $\eps$-free infinitesimal generator:
 	\begin{equation}
 	\MCL_0 = \half a^2(y) \partial_y^2 + b(y) \partial_y.
 	\end{equation}
	To apply the singular perturbation, we assume that $Y^{(1)}$ is ergodic and equipped with a unique invariant distribution $\Phi$. We denote by $\average{\cdot}$ the average with respect to $\Phi$:
	\begin{equation}
	\average{f} = \int f\ud \Phi.
	\end{equation} 
	For further discussion on the model \eqref{p2_eq_Stfull}--\eqref{p2_eq_Ztfull}, including asymptotic results of option pricing as $(\eps, \delta) \to 0$, we refer to \cite{FoPaSiSo:11}.

	Recall from Section~\ref{sec_intro}  the wealth process $X_t^\pi$ associated to the feedback trading strategy $\pi$:
	\begin{equation*}
	\ud X_t^\pi = \pi(t,X_t^\pi,Y_t,Z_t)\mu(Y_t,Z_t)\ud t + \pi(t,X_t^\pi, Y_t,Z_t)\sigma(Y_t,Z_t) \ud W_t,
	\end{equation*}
	and the value function $\Vf(t,x,y,z)$: 
	\begin{equation*}
	\Vf(t,x,y,z)  = \sup_{\pi \in \MCA^{\eps, \delta}} \EE\left[U(X_T^\pi)\vert X_t^\pi = x, Y_t = y, Z_t = z\right],
	\end{equation*}
    Section~\ref{p2_sec_multiheuristic} reviews heuristic expansion results of $\Vf$ studied in \cite[Section 4]{FoSiZa:13}. Before that, we shall detour a bit by reviewing the classical Merton problem, as it plays an important role in the asymptotic analysis as well as in the later proofs. 
	
\subsection{Merton PDE and risk-tolerance function}	
	In Merton's original work \cite{Me:69,Me:71}, both return $\mu$ and volatility $\sigma$ in \eqref{p2_eq_Stfull} are constants. In this case, the wealth process denote by $X_t^\pi$ (with some abuse of notation) becomes
	\begin{equation}
	\ud X_t^\pi = \pi(t, X_t^\pi) \mu \ud t + \pi(t, X_t^\pi) \sigma \ud W_t.
	\end{equation}
	Following the notations in \cite{FoSiZa:13}, we denote by $M(t, x; \lambda)$ the corresponding Merton value function:
	 \begin{equation}\label{def_Merton}
	 M(t,x;\lambda) = \sup_{\pi} \EE[U(X_T^\pi)\vert X_t^\pi = x],
	 \end{equation}
	 where $\lambda$ is the Sharpe-ratio $\lambda = \mu / \sigma$.
	 The reason to show the explicit dependence on $\lambda$ is that  $M(t,x;\lambda)$ is characterized by the nonlinear equation
	 \begin{equation}\label{p1_eq_value}
	 M_t -\frac{1}{2}\lambda^2\frac{M_x^2}{M_{xx}} = 0, \quad M(T,x;\lambda) = U(x).
	 \end{equation}
	 where $\lambda$ appears as a parameter. Later on,  when identifying $\vz$ in \eqref{eq_Vfexpansion}, the notation $M$ will be used repeatedly with different $\lambda$. 
	 
	 The PDE \eqref{p1_eq_value} is obtained by applying dynamic programming principle which gives
	 \begin{equation*}
	 M_t+\sup_{\pi}\left\{\frac{1}{2}\sigma^2\pi^2M_{xx}+\mu\pi M_x\right\}=0,
	 \end{equation*}
	 and then plugging in the candidate of optimal strategy
	 \begin{equation}\label{p1_eq_pistar}
	 \pi^\star(t,x;\lambda)=-\frac{\lambda}{\sigma}\frac{M_x(t,x;\lambda)}{M_{xx}(t,x;\lambda)}.
	 \end{equation}
	 The {verification theorem}  \cite[Chapter 19] {Bj:09} ensures that solving the HJB equation also acts as a sufficient condition for the problem \eqref{def_object}. We next provide some results that are related to the later derivations. The proofs are omitted for the sake of brevity, and we refer readers to \cite[Section~2.1]{FoHu:16} for details.
	 \begin{prop}\label{p1_prop_Merton}
	 	Assume that the utility function $U(x)$ is $C^2(0,\infty)$, strictly increasing, strictly concave, such that $U(0+)$ is finite and satisfies the Inada and asymptotic elasticity conditions:
	 	\begin{equation*}
	 	U'(0+) = \infty, \quad U'(\infty) = 0, \quad \text{AE}[U] := \lim_{x\rightarrow \infty} x\frac{U'(x)}{U(x)} <1,
	 	\end{equation*}
	 	then, the Merton value function $M(t,x;\lambda)$ is strictly increasing, strictly concave in the wealth variable $x$, and decreasing in the time variable $t$. It is $C^{1,2}([0,T]\times \RR^+)$ and is the unique solution to equation \eqref{p1_eq_value}. It is  $C^1$ with respect to $\lambda$, and the optimal portfolio $\pi^\ast$ is given by \eqref{p1_eq_pistar}.
	 	%	\begin{equation}\label{p1_eq_pistar}
	 	%	\pi^\star(t,x;\lambda)=-\frac{\lambda}{\sigma}\frac{M_x(t,x;\lambda)}{M_{xx}(t,x;\lambda)}.
	 	%	\end{equation}
	 \end{prop}

	 Next, we define the {risk-tolerance function} $R(t,x;\lambda)$ associated with the classical Merton value function:
	 \begin{equation}\label{p1_def_risktolerance}
	 R(t,x;\lambda) = -\frac{M_x(t,x;\lambda)}{ M_{xx}(t,x;\lambda)}.
	 \end{equation}
	 It plays an important role in the analysis in later chapters. Note that $R(t,x;\lambda)$ is well-defined, continuous, strictly positive due to the regularity, strict concavity and monotonicity of $M (t,x;\lambda)$. Following the notation in \cite{FoSiZa:13}, we define the differential operators in terms of $R(t,x;\lambda)$
	 \begin{align}\label{p1_def_dk}
	 D_k(\lambda) &= R(t,x;\lambda)^k \partial_x^k, \qquad k = 1,2, \cdots,\\
	 \Ltx(\lambda) &= \partial_t + \frac{1}{2}\lambda^2D_2(\lambda) + \lambda^2D_1(\lambda).\label{p1_def_ltx}
	 \end{align}
	 Note that the coefficients of $\Ltx(\lambda)$ depend on $R(t,x;\lambda)$, and consequently on $M(t,x;\lambda)$. Thus, the Merton PDE \eqref{p1_eq_value} can be rewritten in a ``linear'' manner
	 \begin{align}\label{p1_eq_mertonlinear}
	 &\Ltx(\lambda)M(t,x;\lambda) = 0,
	 \end{align}
	 and this PDE possesses a unique nonnegative solution.
	 
	 \begin{prop}\label{p1_prop_ltxunique}
	 	Let $\Ltx(\lambda)$ be the operator defined in \eqref{p1_def_ltx}, and assume that the utility function $U(x)$ satisfies the conditions in Proposition \ref{p1_prop_Merton}, then
	 	\begin{equation}\label{p1_def_ltxpde}
	 	\Ltx(\lambda)u(t,x;\lambda) = 0,\quad
	 	u(T,x;\lambda) = U(x),
	 	\end{equation}
	 	has a unique nonnegative solution. Consequently, this PDE with zero terminal condition possesses only trivial solution. 
	 \end{prop}
	
	\subsection{Multiscale asymptotic expansions}\label{p2_sec_multiheuristic}
	We now summarize some existing heuristics derived in \cite[Section~4]{FoSiZa:13}. By dynamic programing, $\Vf$ solves the following HJB equation:
	\begin{equation}\label{p2_eq_VfPDE}
	\left(\partial_t + \frac{1}{\eps} \MCL_0 + \delta \MCM_2 + \sqrt{\frac{\delta}{\eps}}\MCM_3\right)\Vf   - \frac{\left(\lambda(y,z)\Vf_x + \frac{\rho_1a(y)}{\sqrt\eps} \Vf_{xy} + \sqrt{\delta}\rho_2g(z)\Vf_{xz}\right)^2}{2\Vf_{xx}} = 0,
	\end{equation}
	with the candidate of the optimal strategy
	\begin{equation*}
	\pi^\star = -\frac{\lambda(y,z)\Vf_x}{\sigma(y,z)\Vf_{xx}} - \frac{\rho_1a(y)\Vf_{xy}}{\sqrt\eps\sigma(y,z)\Vf_{xx}}- \frac{\sqrt{\delta} \rho_2g(z)\Vf_{xz}}{\sigma(y,z)\Vf_{xx}},
	\end{equation*}
	where $\MCM_3$ is defined as:
	\begin{equation}
	\MCM_3 = \rho_{12}a(y)g(z)\partial_y\partial_z,
	\end{equation}
	and $\lambda(y,z) = \mu(y,z) / \sigma(y,z)$ is the Sharpe ratio function.
	
	In general, $\Vf$ is only identified as the viscosity solution of the above HJB equation \cite[Section 4]{Phbook:09}. However, to apply asymptotic derivations,  \cite{FoSiZa:13} assume that $\Vf$ is smooth in every variable, strictly increasing, strictly concave in  the wealth argument $x$ for each $(y,z)$ in $\RR^2$ and $t \in [0,T)$, and is the unique classical solution to \eqref{p2_eq_VfPDE}. We emphasize that  results in this paper do not rely on the regularity of $\Vf$, as we will work with $\Vfl$ defined in \eqref{p2_def_Vfl}, which will be classical solution of the linear PDE \eqref{p2_eq_Vfl}.
	
	The multiscale expansion consists of constructing a power series of $\delta$ for $\Vf$:
	\begin{equation}
	\Vf = V^{\eps,0} + \sqrt{\delta} V^{\eps,1} + \cdots
	\end{equation}
	 and then a power series in $\eps$ for each term $V^{\eps, k}$:
	 \begin{equation}
	 V^{\eps,k} = v^{(0,k)} + \sqrt{\eps} v^{(1,k)} + \eps v^{(2,k)} + \cdots, \quad \forall k \in \NN.
	 \end{equation}
	  At each step, the coefficients $V^{\eps, k}$ or $v^{(j,k)}$ are identified by substituting the expansion into the corresponding equation and collecting terms of different orders. Because the whole analysis will be performed on $\Vfl$ again in Section~\ref{p2_sec_multipz}, we decide to skip the derivation here and jump to the results. The combined expansion in slow and fast scale of $\Vf$ is of the following form:
	\begin{equation}\label{p2_eq_Vfexpansion}
	\Vf = \vz + \sqrt{\eps}\voz + \sqrt{\delta}\vzo + \eps\vtz + \delta \vztw + \sqrt{\eps\delta}\voo + \cdots,
	\end{equation}
	where the superscript of $v$ corresponds to the power in $\sqrt\eps$ and $\sqrt\delta$ and where  $v^{(0,0)}$ is rewritten as $v^{(0)}$.  Formulations about $\vz$, $\voz$ and $\vzo$ are given as follows. 
	\begin{enumerate}[(i)]
		
		\item The \emph{leading order term} $\vz$ is defined as the solution to the Merton PDE associated with the ``averaged'' Sharpe ratio $\aves(z)=\sqrt{\average{\lambda^2(\cdot,z)}}$,
		\begin{equation}\label{p2_eq_vzf}
		\vz_t - \half\aves^2(z)\frac{\left(\vz_x\right)^2}{\vz_{xx}} = 0, \quad \vz(T,x,z) = U(x).
		\end{equation}
		Since it  possesses a unique solution, we have
		\begin{equation}\label{p2_eq_vzfvsmerton}
		\vz(t,x,z) = M(t,x;\aves(z)).
		\end{equation}
		Accordingly, the version of $D_k(\lambda)$ that will be used in the sequel is $D_k(\aves) = R(t,x;\aves(z))^k \partial_x^k$ under the multiscale stochastic environment, and we shall use $D_k$ for brevity (omitting the argument $\aves$).
		
		\item The \emph{first order correction} in the fast variable $\voz$ is defined as the solution to the linear PDE:
		\begin{equation}\label{p2_eq_vozf}
		\voz_t + \half\aves^2(z)\left(\frac{\vz_x}{\vz_{xx}}\right)^2\voz_{xx}-\aves^2(z)\frac{\vz_x}{\vz_{xx}}\voz_{x} = \half \rho_1B(z)D_1^2\vz, \quad \voz(T,x,z) = 0,
		\end{equation}
		which admits a unique solution. Then $\voz$ is explicitly given in terms of $\vz$ by
		\begin{equation}\label{p2_def_vozf}
		\voz(t,x,z) = -\half(T-t)\rho_1B(z)D_1^2\vz(t,x,z),
		\end{equation}
		where
		\begin{equation}
		B(z) = \average{\lambda(\cdot,z)a(\cdot)\partial_y\theta(\cdot,z)}, \quad \MCL_0\theta(y,z) = \lambda^2(y,z) - \aves^2(z).
		\end{equation}
		Note that in the solution $\theta(y,z)$ to the above Poisson equation, the variable $z$ can be treated as a parameter.
		
		\item The \emph{first order correction} in the slow variable $\vzo$ is defined as the solution to the linear PDE:
		\begin{align}\label{p2_eq_vzof}
		&\vzo_t + \half\aves^2(z)\left(\frac{\vz_x}{\vz_{xx}}\right)^2\vzo_{xx}-\aves^2(z)\frac{\vz_x}{\vz_{xx}}\vzo_x - \rho_2\widehat \lambda(z)g(z)\frac{\vz_x}{\vz_{xx}}\vz_{xz} =0, \\		
		&\vzo(T,x,z) = 0,
		\end{align}
		which has a unique solution, and where $\widehat \lambda(z)$ is given by
		\begin{equation}
		\widehat \lambda(z) = \average{\lambda(\cdot,z)}.	
		\end{equation}
		
		\item By the``Vega-Gamma'' relation, the $z$-derivative of the leading order term $\vz$ satisfies
		\begin{equation}\label{p2_eq_vegagamma}
		\vz_z(t,x,z) = (T-t)\aves(z)\aves'(z)D_1\vz(t,x,z),
		\end{equation}
		and $\vzo$ can be expressed in terms of $\vz$ by
		\begin{multline}\label{p2_eq_vzo}
		\vzo(t,x,z) = \half(T-t)\rho_2\widehat{\lambda}(z)g(z)D_1\vz_z(t,x,z) \\ =\half(T-t)^2\rho_2\widehat{\lambda}(z)\aves(z)\aves'(z)g(z)D_1^2\vz(t,x,z).
		\end{multline}
	\end{enumerate}

Note that the uniqueness in (i)--(iii) follows from Proposition~\ref{p1_prop_Merton} and \ref{p1_prop_ltxunique}.

	\subsection{Model assumptions and preliminary estimates}\label{sec_modelassump}
	There are two sets of assumptions needed for results presented in this paper: one about the state process, and one on the general utility. The first set is basically the combination of Assumption~2.12 in our previous work \cite{FoHu:16} considering solely the slow factor, and Assumption~2.4 in the fast case \cite{Hu:18}, except that formulations based on $\lambda(z)$ (slow case) and $\aves$ (fast case) are all shifted to the multiscale case $\aves(z)$. The second set extends Assumption~2.5 in \cite{FoHu:16} by requiring more regularity of $U(x)$ and more boundedness constraints on the risk-tolerance $R(x) = -U'(x)/U''(x)$. For completeness, we next present them in details.

	For fixed $(t,z)$, we observe that, $\vz(t,x,z) = M(t,x; \aves(z))$ is a concave function that has a linear upper bound. In fact, for $t=0$, there exists a function $\overline G(z)$, so that $$\vz(0,x,z) \leq \overline G(z) + x, \;\forall (x,z) \in \RR^+\times \RR.$$ Let $X_t^\pz$ be the the wealth process following $\pz$, and define $\pz$ in terms of model parameters and the zeroth order term $\vz(t,x,z)$:
	\begin{equation} \label{p2_eq_pzfull}
	\pz = -\frac{\lambda(y,z)}{\sigma(y,z)}\frac{\vz_x(t,x,z)}{\vz_{xx}(t,x,z)}.
	\end{equation}
	%The following assumption is about $(S_t,Y_t, Z_t, X_t^\pz)$ and the function $\overline G$.

	\begin{assump}\label{p2_assump_valuefuncadd} We make the following assumptions on the state processes $(S_t,Y_t, Z_t, X_t^\pz)$:
		\begin{enumerate}[(i)]
			\item\label{p2_assump_valuefuncSZadd} For any starting points $(s, y, z)$ and fixed $(\eps, \delta)$, the system of SDEs \eqref{p2_eq_Stfull}--\eqref{p2_eq_Ytfull}--\eqref{p2_eq_Ztfull} has a unique strong solution $(S_t, Y_t, Z_t)$. The function $g(z)$ is in $C^2(\RR)$, and $\lambda(y,z)$ is in $C^3(\RR)$ in the $z$-variable.
			The coefficients $g(z)$, $c(z)$, $a(y)$ $\lambda(y, z)$ as well as their derivatives $g'(z)$, $g''(z)$, $\lambda_z(y,z)$, $\lambda_{zz}(y,z)$, and $\lambda_{zzz}(y, z)$ are at most polynomially growing.
			
			\item The process $Y^{(1)}$ with infinitesimal generator $\MCL_0$ is ergodic with a unique invariant distribution $\Phi$, and admits moments of any order uniformly in $t \leq T$:
			\begin{equation}
			\sup_{t \leq T} \left\{\EE\abs{Y_t^{(1)}}^k\right\} \leq C(T,k).
			\end{equation}
			The solution $\phi(y,z)$ of the Poisson equation (in $y$) $\MCL_0\phi(y,z) = \ell(y,z)$ is assumed to be polynomial for polynomial (in $y$) function $\ell(y,z)$.
			\item\label{p2_assump_valuefuncZmomentadd} The process $Z^{(1)}$ with infinitesimal generator $\MCM_2$ defined in \eqref{p1_eq_Mgenerator} admits moments of any order uniformly in $t \leq  T$:
			\begin{equation}
			\sup_{t \leq T}\left\{ \EE\abs{Z_t^{(1)}}^k\right\} \leq C(T,k).
			\end{equation}
			\item\label{p2_assump_valuefuncGadd}  The process $\overline G(Z_\cdot)$ is in $L^2([0,T]\times \Omega)$ uniformly in $\delta$, i.e.,
			\begin{equation}\label{p2_assump_Gzadd}
			\EE_{(0,z)}\left[\int_0^T \overline G^2(Z_s) \ud s\right] \leq C_1(T,z),
			\end{equation}
			where $C_1(T,z)$ is independent of $\delta$ and $Z_s$ follows \eqref{p2_eq_Ztfull} with $Z_0 = z$.
			\item\label{p2_assump_valuefuncXadd} The wealth process $X_\cdot^\pz$ {stays nonnegative, namely, $\pz \in \MCA^{\eps,\delta}(t,x,y, z)$ $\forall 0<\eps,\delta\leq 1$}. {Moreover, it}
			is in  $L^2([0,T]\times \Omega)$ uniformly in $(\eps, \delta)$ , i.e.,
			\begin{equation}\label{p2_assump_Xsquareadd}
			\EE_{(0,x,y, z)}\left[\int_0^T \left(X_s^\pz\right)^2 \ud s\right] \leq C_2(T,x,y, z),
			\end{equation}
			where $C_2(T,x,y, z)$ is independent of $(\eps, \delta)$.
		\end{enumerate}
	\end{assump}
	
	\begin{lem}\label{p2_lem_unibddadd}
		Under Assumption~\ref{p2_assump_valuefuncadd}\eqref{p2_assump_valuefuncGadd}-\eqref{p2_assump_valuefuncXadd}, the process $\vz(\cdot, X_\cdot^\pz,Z_\cdot)$ is in $L^2([0,T]\times \Omega)$ uniformly in $(\eps,\delta)$, i.e. $\forall (t,x,y, z) \in [0,T]\times\RR^+\times \RR\times \RR$, we have
		\begin{equation}
		\EE_{(t,x,y, z)}\left[\int_t^T \left(\vz(s,X_s^\pz, Z_s)\right)^2 \ud s\right] \leq C_3(T,x,y,z),
		\end{equation}
		where $\vz(t,x,z)$ is defined in Section \ref{p2_sec_multiheuristic} and satisfies $\vz(t,x,z) = M(t,x;\aves(z))$.
	\end{lem}
	\begin{proof}
		The proof follows the argument in \cite[Lemma~2.15]{FoHu:16}.
	\end{proof}		
	
		\begin{assump}\label{assump_U}
			We make the following assumptions on $U(x)$:
			\begin{enumerate}[(i)]
				\item\label{p1_assump_Uregularity}  U(x) is $C^9(0,\infty)$, strictly increasing, and strictly concave and satisfies the following conditions (Inada and asymptotic elasticity):
				\begin{equation}\label{p1_eq_usualconditions}
				U'(0+) = \infty, \quad U'(\infty) = 0, \quad \text{AE}[U] := \lim_{x\rightarrow \infty} x\frac{U'(x)}{U(x)} <1.
				\end{equation}
				\item\label{p1_assump_Ubddbelow}U(0+) is finite. Without loss of generality, we assume U(0+) = 0.
				\item\label{p1_assump_Urisktolerance} Assume the risk tolerance $R(x) := -U'(x) / U''(x)$ satisfies $R(0) = 0$, strictly increasing, $R'(x) < \infty$ on $[0,\infty)$, and there exists $K\in\RR^+$, such that for $x \geq 0$, and $ 2\leq i \leq 7$,
				\begin{equation}\label{p1_assump_Uiii}
				\abs{\partial_x^{(i)}R^i(x)} \leq K.
				\end{equation}
				\item\label{p1_assump_Ugrowth} Define the inverse function of the marginal utility $U'(x)$ as $I: \RR^+ \to \RR^+$, $I(y) = U'^{(-1)}(y)$, and assume that, for some positive $\alpha$, $I(y)$ satisfies the polynomial growth condition:
				\begin{equation}\label{p1_cond_I}
				I(y) \leq \alpha + \kappa y^{-\alpha}.
				\end{equation}
			\end{enumerate}
		\end{assump}

Note that Assumption~\ref{assump_U}\eqref{p1_assump_Ubddbelow} is a sufficient condition, and excludes the case $U(x) = \frac{x^\gamma}{\gamma}$, for $\gamma < 0$, and $U(x) = \log(x)$. However, all theorem in this paper still hold under minor modifications to the proof.  Further discussion about the above assumptions, regarding examples, restrictiveness and implication can be found in \cite[Section~2.3]{FoHu:16}.

We next give some estimate of the risk-tolerance function \eqref{p1_def_risktolerance}. Since in the multiscale regime, the zeroth order term $\vz(t,x,z)$ is identified as $M(t,x;\aves(z))$, see equation \eqref{p2_eq_vzfvsmerton}, the notation of the risk-tolerance function is changed accordingly to $R(t,x;\aves(z))$ with
\begin{equation}
R(t,x;\aves(z)) := -\frac{\vz(t,x,z)}{\vz(t,x,z)} = -\frac{M_x(t,x;\aves(z))}{M_{xx}(t,x;\aves(z))}
\end{equation}
to emphasis the dependence of $\aves(z)$. 
\begin{prop}\label{p2_prop_risktoleranceadd}
	Under Assumption~\ref{assump_U} of the general utility, the risk-tolerance $R(t,x;\aves(z))$ function satisfies the following:  $\exists K_j >0$ for $0 \leq j  \leq  6$, such that $\forall (t,x,\aves(z)) \in [0,T) \times \RR^+ \times \RR$,
	\begin{equation}
	\abs{R^j(t,x;\aves(z))(\partial_x^{(j+1)}R(t,x;\aves(z)))} \leq K_j.
	\end{equation}
	Or equivalently, $\forall 1 \leq j \leq 7$, there exists $\widetilde K_j >0$, such that $\forall (t,x,z) \in [0,T) \times \RR^+ \times \RR$,
	\begin{equation}
	\abs{\partial_x^{(j)} R^j(t,x;\aves(z))} \leq \widetilde K_j.
	\end{equation}
	Moreover, the following quantities are uniformly bounded: $RR_{xxz}$, $R^2R_{xxxz}$, $R_{xzz}$, $RR_{xxzz}$ and $R^2R_{xxxzz}$.
\end{prop}
\begin{proof}
	The first part extends results of \cite[Proposition~3.5]{FoHu:16}, and the proof is essentially repeating the argument therein  for the case $j = 5$ and 6 by using the comparison principle of heat equations. The proof of the second part consists of successively differentiating the ``Vega-Gamma'' relation in \eqref{p2_eq_vegagamma}, and repeatedly using the concavity of $\vz$, the results in the first part, and Propositions~3,5, 3.6 and 3.7 in \cite{FoHu:16} with $\lambda$ or $\lambda(z)$ replaced by $\aves(z)$. For the sake of simplicity, we omit this lengthy, tedious but straightforward derivation.
\end{proof}

\section{$\pz$-Portfolio performance under multiscale regime}\label{p2_sec_multipz}
	Recalling  the strategy $\pz$ defined in terms of model parameters and the zeroth order term $\vz(t,x,z)$ in \eqref{p2_eq_pzfull}:
	\begin{equation}
	\pz = -\frac{\lambda(y,z)}{\sigma(y,z)}\frac{\vz_x(t,x,z)}{\vz_{xx}(t,x,z)} := \frac{\lambda(y,z)}{\sigma(y,z)}R(t,x;\aves(z)),
	\end{equation}
	and assuming $\pz$ is admissible, we shall give its performance in this section. More precisely, let $X^\pz$  be the wealth process following $\pz$
	\begin{equation}
	\ud X_t^\pz = \mu(Y_t,Z_t)\pz(t,X_t^\pz,Y_t,Z_t)\ud t + \sigma(Y_t,Z_t)\pz(t,X_t^\pz,Y_t,Z_t) \ud W_t,
	\end{equation}
	then we aim at finding the rigorous approximation of the associated value function:
	\begin{equation}\label{p2_def_Vfl}
	\Vfl := \EE\left\{U(X_T^\pz)\Big\vert X_t^\pz = x, Y_t = y, Z_t = z\right\},
	\end{equation}
	with the general utility $U$ satisfying Assumption~\ref{assump_U}. The estimation result regarding $\Vfl$ has been summarized in Theorem~\ref{p2_Thm_full}, and we present the proof in the next subsection.

\subsection{Proof of Theorem~\ref{p2_Thm_full}}

The proof is split into two steps: firstly, we propose the expansion form of $V^{\pz, \eps, \delta}$: 
\begin{equation}\label{eq_Vfl}
V^{\pz, \eps, \delta} = \vzl + \sqrt{\eps}\vozl + \sqrt{\delta}\vzol + \cdots,
\end{equation}
and identifying $\vzl$, $\vozl$ and $\vzol$ properly. To this end, we write down the PDE satisfied by $V^{\pz, \eps, \delta}$, perform regular perturbations in the slow parameter $\delta$, and then singular perturbations in the fast parameter $\eps$. Note that this technique has been used in the linear pricing problem with two factor models, for instance, see \cite[Section 4]{FoPaSiSo:11}. 
Secondly, we justify the ``$\cdots$'' part in \eqref{eq_Vfl} is of order $\MCO(\eps + \delta)$, to complete the proof. 

\medskip
\noindent\textbf{Step1: Heuristic derivations.}	By the martingale property, we note that $\Vfl$ satisfies the following linear PDE
		\begin{align}\label{p2_eq_Vfl}
		\left(\MCL_2 +\frac{1}{\sqrt{\eps}}\MCL_1 +  \frac{1}{\eps}\MCL_0 + \delta \MCM_2 + \sqrt{\delta}\MCM_1 + \frac{\sqrt\delta}{\sqrt{\eps}}\MCM_3\right)\Vfl &=0,  \\
		\Vfl(T,x,y,z) &= U(x),
		\end{align}
		where the operators $\MCL_i$ and $\MCM_i$ are defined by:
		\begin{align}
		\MCL_0 &= b(y)\partial_y + \half a^2(y)\partial^2_{y}, &\MCL_1 &= \rho_1a(y)\sigma(y,z)\pz\partial_x\partial_y,&\\
		\MCL_2 &= \partial_t + \mu(y,z)\pz\partial_x + \half\sigma^2(y,z)\left(\pz\right)^2\partial^2_{x}, &\MCM_1 &= \rho_2\sigma(y,z)g(z)\pz\partial_x\partial_z,&\\
		\MCM_2 &= c(z)\partial_z + \half g^2(z)\partial^2_z, &\MCM_3 &= \rho_{12}a(y)g(z)\partial_y\partial_z.&
		\end{align}
		
		The strategy is to expand the value function first in the slow parameter $\delta$:
		\begin{equation}\label{p2_eq_Vflexpansion}
		\Vfl = \Vflz + \sqrt{\delta}\Vflo + \cdots,
		\end{equation}
		and identify the effective equations for $\Vflz$, $\Vflo$, 
		and then to expand $\Vflz$ and $\Vflo$ in the fast parameter $\eps$
		\begin{align*}
		&\Vflz = \vzl + \sqrt{\eps}\vozl + \eps\vtzl + \eps^{3/2}\vthzl + \cdots,\\
		&\Vflo = \vzol + \sqrt{\eps}\vool +\eps\vtol  \cdots.
		\end{align*}
		Again the superscript $(i,j)$ of $v^\pz$ indicates the power in $\sqrt{\eps}$ and $\sqrt{\delta}$ respectively, and $(0,0)$ is reduced to $(0)$ for being consistent with the notations in \cite{FoSiZa:13}.
		By letting $\delta = 0$, we deduce
		\begin{equation}
		\left(\MCL_2 + \frac{1}{\sqrt{\eps}}\MCL_1 + \frac{1}{\eps}\MCL_0 \right)\Vflz = 0, \quad \Vflz(T,x,y,z) = U(x).
		\end{equation}
		This is actually equation (17) in \cite{Hu:18} expect that $\lambda(y)$ is replaced by $\lambda(y,z)$ to take the slow factor $Z_t$ into consideration. However, $z$ can be viewed as a parameter in $\Vflz$, as there is no $z$-derivatives in the above PDE. Consequently,  the derivation and reasoning in \cite[Section~III.A]{Hu:18} can be applied, and we deduce
		\begin{align}\label{p2_eq_vzl}
		&\vzl = \vz = M(t,x,\aves(z)),\\  \label{p2_eq_vtzl}
		& \vtzl = -\half\theta(y,z)D_1\vz + C_1(t,x,z)\\ \label{p2_eq_vozl}
		&\vozl = \voz = -\half(T-t)\rho_1B(z)D_1^2\vz,\\  \label{p2_eq_vthzl}
		& \vthzl =  \frac{1}{2}(T-t)\theta(y,z)\rho_1B(z)\left(\frac{1}{2}D_2 + D_1\right)D_1^2\vz + \frac{1}{2} \rho_1\theta_1(y,z)D_1^2\vz + C_2(t,x,z),
		\end{align}
		where $\theta_1(y)$ is the solution to the ODE:
		\begin{equation}
		\MCL_0\theta_1(y) = \lambda(y,z)a(y)\partial_y\theta(y,z) - \average{\lambda(\cdot,z)a(\cdot)\partial_y\theta(\cdot,z)},
		\end{equation}
		and $C_i(t,x,z)$ are some constant of integration in $y$, that may depend on $(t,x,z)$, for $i = 1, 2$.
		
		Next, we go back to equation \eqref{p2_eq_Vfl} and derive the PDE for $\Vflo$ by collecting terms of order $\sqrt{\delta}$,
		\begin{equation}
		\left(\MCL_2 + \frac{1}{\sqrt{\eps}}\MCL_1 + \frac{1}{\eps}\MCL_0\right)\Vflo + \left(\MCM_1 + \frac{1}{\sqrt{\eps}}\MCM_3\right)\Vflz = 0, \quad \Vflo(T,x,y,z) = 0.
		\end{equation}
		Observing that $\MCM_3$ takes derivatives in $y$, and in the expansion of $\Vfl$, the first two terms $\vzl$ and $\vozl$ are independent of $y$ (cf. \eqref{p2_eq_vzl} and \eqref{p2_eq_vozl}), one has
		\begin{equation}
		\frac{1}{\sqrt{\eps}}\MCM_3 \Vfl = \sqrt{\eps} \vtzl + \eps \vthzl + \cdots. 
		\end{equation}
		Now, by collecting $\eps^{-1}$ terms and $\frac{1}{\sqrt{\eps}}$ terms and noticing that there is no $y$-derivative in the equations, we could make the choices that $\vzol = \vzol(t,x,z)$ and $\vool = \vool(t,x,z)$, i.e. they are independent of $y$. Collecting terms of order one forms a Poisson equation for $\vtol$,
		\begin{equation}\label{p2_eq_vtolvzol}
		\MCL_0\vtol + \MCL_2\vzol + \MCM_1\vzl = 0, \quad \vzol(T,x,z) = 0.
		\end{equation}
		and yields the following solvability condition for $\vzol$
		\begin{align}
		&\vzol_t + \half\aves^2(z)\left(\frac{\vz_x}{\vz_{xx}}\right)^2\vzol_{xx} - \aves^2(z)\frac{\vz_x}{\vz_{xx}}\vzol_x+\rho_2\widehat{\lambda}(z)g(z)R\vz_{xz} = 0,\\ 
		&\vzol(T,x,z) = 0.
		\end{align}
		where we have used the relation $\vzl = \vz$ (cf. \eqref{p2_eq_vzl}) . This is exactly equation \eqref{p2_eq_vzof}, by its uniqueness,  we obtain
		\begin{align}\label{p2_eq_vzol}
		\vzol = \vzo &= \half(T-t)\rho_2\widehat{\lambda}(z)g(z)D_1\vz_z(t,x,z) \\ &=\half(T-t)^2\rho_2\widehat{\lambda}(z)\aves(z)\aves'(z)g(z)D_1^2\vz(t,x,z).
		\end{align}
		Plugging it back to equation \eqref{p2_eq_vtolvzol}, and solving for $\vtol$ gives
		\begin{equation}\label{p2_eq_vtoltmp}
		\vtol = -\theta(y,z)\left(\half D_2 + D_1\right)\vzo - \theta_2(y,z)\rho_2g(z)D_1\vz_z + C_3(t,x,z),
		\end{equation}
		where $\theta_2(y,z)$ is the solution to the ODE
		\begin{equation}
		\MCL_0 \theta_2(y,z) = \lambda(y,z) - \widehat \lambda(z),
		\end{equation}
		and $C_3(t,x,z)$ is some `constant' in $y$. 
		To further express $\vtol$ in terms of $\vz$ solely, we use the expression \eqref{p2_eq_vzo} of $\vzo$ and obtain
		\begin{align}\label{p2_eq_vtol}
		\vtol &= -\theta(y,z)\half(T-t)^2\rho_2\widehat{\lambda}\aves\aves'g\left(\half D_2 + D_1\right)D_1^2\vz \\
		&\quad - \theta_2(y,z)\rho_2g(z)(T-t)\aves\aves'D_1^2\vz + C_3(t,x,z).
		\end{align}

Till now, desired terms are all identified including $\vzl, \vtzl, \vozl, \vthzl, \vzol$ and $\vtol$, and we will move on to the justification of the above derivation.

\medskip

\noindent\textbf{Step2: Expansion justification.}
		To validate the above formal derivation, we at least need to show the residual function $E(t,x,y,z)$ is of order higher than $\sqrt\eps + \sqrt\delta)$. To this end, we first analyze an auxiliary residual function $\widetilde E(t,x,y,z)$ defined by
		\begin{equation}
		\widetilde E(t,x,y,z)= \Vfl - \vz - \sqrt{\eps}\voz - \eps \vtzl - \eps^{3/2}\vthzl - \sqrt{\delta}\vzo - \eps\sqrt{\delta} \vtol,
		\end{equation}
		with these functions given in \eqref{p2_eq_vzl}, \eqref{p2_eq_vtzl}, \eqref{p2_eq_vozl}, \eqref{p2_eq_vthzl}, \eqref{p2_eq_vzol} and \eqref{p2_eq_vtol}, respectively. We take $C_i(t,x,z) \equiv 0$, $i = 1, 2, 3$ in the relevant terms. Straight forward computation gives
		\begin{align*}
		&\left(\MCL_2 +\frac{1}{\sqrt{\eps}}\MCL_1 +  \frac{1}{\eps}\MCL_0 + \delta \MCM_2 + \sqrt{\delta}\MCM_1 + \frac{\sqrt\delta}{\sqrt{\eps}}\MCM_3\right)\widetilde E \\
		&\qquad\quad + \MCL_2\left(\eps\vtzl + \eps^{3/2}\vthzl + \eps\sqrt{\delta}\vtol\right) + \MCL_1\left(\eps\vthzl + \sqrt{\eps\delta}\vtol\right) \\
		&\qquad\quad + \sqrt{\delta}\MCM_3\left(\sqrt{\eps}\vtzl + \eps\vthzl + \sqrt{\eps\delta}\vtol\right)\\
		&\qquad\quad + \delta\MCM_2\left(\vz + \sqrt{\eps}\voz + \eps \vtzl + \eps^{3/2}\vthzl + \sqrt{\delta}\vzo + \eps\sqrt{\delta} \vtol\right)\\
		&\qquad\quad + \sqrt{\delta}\MCM_1\left(\sqrt{\eps}\voz + \eps\vtzl + \eps^{3/2}\vthzl + \sqrt{\delta}\vzo + \eps\sqrt{\delta}\vtol\right) = 0, \\
		& \widetilde E(T,x,y,z) = -\eps\vtzl(T,x,y,z) - \eps^{3/2}\vthzl(T,x,y,z).
		\end{align*}
		
		Noticing that $\MCL_2 +\frac{1}{\sqrt{\eps}}\MCL_1 +  \frac{1}{\eps}\MCL_0 + \delta \MCM_2 + \sqrt{\delta}\MCM_1 + \frac{\sqrt\delta}{\sqrt{\eps}}\MCM_3$ is the infinitesimal generator of the processes $(X_t^\pz, Y_t, Z_t)$, and using the bound estimates which we will show next, we have the following Feynman--Kac formula:
		\begin{align}
		\widetilde E(t,x,y,z) &= \eps \EE_{(t,x,y,z)}\left[\int_t^T \MCR^1 (s,X_s^\pz, Y_s,Z_s) \ud s \right] +  \delta \EE_{(t,x,y,z)}\left[\int_t^T \MCR^2  (s,X_s^\pz, Y_s,Z_s) \ud s \right] \nonumber\\
		& \quad + \sqrt{ \eps\delta} \EE_{(t,x,y,z)}\left[\int_t^T \MCR^3  (s,X_s^\pz, Y_s, Z_s) \ud s \right]  - \eps\EE_{(t,x,y,z)}\left[\vtzl(T,X_T^\pz, Y_T,Z_T) \right] \nonumber\\
		& \quad - \eps^{3/2}\EE_{(t,x,y,z)}\left[\vthzl(T,X_T^\pz,Y_T,Z_T)\right],\label{p2_eq_Eflprob}
		\end{align}
		and obtain the desired results for $\widetilde E \sim \MCO(\eps + \delta)$, where $\MCR^i$ are defined as:
		\begin{align}
		\MCR^1& = \MCL_2\left(\vtzl + \sqrt{\eps}\vthzl + \sqrt{\delta}\vtol\right)+\MCL_1\vthzl + \sqrt{\delta}\MCM_1\vtzl\\
		\MCR^2& = \MCM_2\left(\vz + \sqrt{\eps}\voz + \eps \vtzl + \eps^{3/2}\vthzl + \sqrt{\delta}\vzo + \eps\sqrt{\delta} \vtol\right) + \MCM_1\vzo\\
		\MCR^3 &= \MCL_1\vtol + \MCM_3\left(\vtzl + \sqrt{\eps}\vthzl + \sqrt{\delta}\vtol\right) \\
		& \quad + \MCM_1\left(\voz + \eps\vthzl  + \sqrt{\eps\delta}\vtol\right).
		\end{align}
	
		We now present the bound estimates of the expectations in \eqref{p2_eq_Eflprob}.
		Straightforward computation shows that each expectation term $\EE_{(t,x,y,z)}\left[\int_t^T \MCR^i_s\ud s \right]$ is a sum of  integrals of the following form:
		\begin{equation}\label{eq_polydvzform}
		\EE_{(t,x,y,z)}\left[\int_t^T h(Y_s,Z_s)\;\MCD\vz(s,X_s^\pz,Z_s)\ud s \right],
		\end{equation}
		where $h(y,z)$ is mostly polynomially growing, and $\MCD\vz$ takes derivatives of $\vz$. According to different operators, the derivatives are
		\begin{align}\label{eq_derivativeLi}
		&\MCL_i: D_1^2\vz, D_2D_1\vz, D_1^3\vz, D_2D_1^2\vz,D_1^4\vz,D_2D_1^3\vz,D_1D_2D_1^2\vz, D_2^2D_1^2\vz\\
		&\MCM_3: \partial_z D_1\vz, \partial_z D_1^2\vz, \partial_z D_1^3\vz, \partial_z D_2D_1^2\vz, \\
		&\MCM_2: \partial_z\vz, \partial_z^2\vz, \partial_z^2D_1\vz, \partial_z^2D_1^2\vz, \partial_z^2D_1^3\vz, \partial_z^2D_2D_1^2\vz, \text{ plus all terms in } \MCM_3\\
		&\MCM_1: D_1\partial_zD_1\vz, D_1\partial_zD_1^2\vz,  D_1\partial_zD_1^3\vz,  D_1\partial_zD_2D_1^2\vz.\label{eq_derivativeM1}
		\end{align}
	   A repeated use of the concavity of $\vz$, Propositions~3.7 in \cite{FoHu:16} and Proposition~\ref{p2_prop_risktoleranceadd} guarantees that $\MCD\vz$ is bounded by a multiple in $z$ of $\vz$, namely, for any $\MCD\vz$ taking the above form, we have
		\begin{equation}\label{eq_kvzform}
		\abs{\MCD\vz(t,x,z)} \leq k(z)\vz(t,x,z), %\text{ for some non-negative function } k(z).
		\end{equation}
		for some non-negative and at most polynomially growing function $k(z)$.  
		
		For clarity, we consider the term $\MCL_2 \vtzl$ in $\MCR^1$ to illustrate the  above procedure. By definition, one has
				\begin{align*}
				\MCL_2 \vtl &= \left(\partial_t + \lambda^2(y,z) R(t,x;\aves(z))\partial_x + \half \lambda^2(y,z)R^2(t,x;\aves(z)) \partial_x^2\right)(-\half \theta(y,z) D_1 \vz)\\
				& = -\half \theta(y,z) \left((\lambda^2(y,z) - \aves^2(z))R(t,x;\aves(z))\partial_x + \half (\lambda^2(y,z) - \aves^2(z))R^2(t,x;\aves(z))\partial_x^2\right)D_1\vz \\
				& = -\half\theta(y,z) (\lambda^2(y,z) - \aves^2(z))(D_1^2 \vz + \half D_2D_1\vz),
				\end{align*}
				where we have used the relation $\Ltx(\aves(z))D_1 \vz = D_1\Ltx(\aves(z))\vz = 0$.
				Under Assumption~\ref{p2_assump_valuefuncadd}, the function $\theta(y,z)(\lambda^2(y,z) - \aves^2(z))$ is at most polynomially growing, thus $\MCL_2\vtl$ is of the form \eqref{eq_polydvzform}. The other terms in $\MCR^i$, for $i = 1, 2, 3$, follow by a similar reasoning, thus are omitted. To illustrate that all derivatives in \eqref{eq_derivativeLi}--\eqref{eq_derivativeM1} can be bounded as in \eqref{eq_kvzform}, we consider $\partial_zD_1^2\vz$ as an example:
		\begin{align*}
		\abs{\partial_z D_1^2\vz} = \abs{\partial_z R(R_x -1)\vz} \leq \abs{R_z(R_x-1)\vz_x} + \abs{R R_{xz}\vz_x} + \abs{R(R_x-1)\vz_{xz}},
		\end{align*}
		where we have omitted the arguments $(t,x; \aves(z))$ for $R$, and $(t,x, z)$ for $\vz$.
		Proposition~\ref{p2_prop_risktoleranceadd} gives $\abs{R_x} \leq K_0$ and $\abs{R} \leq K_0x$ since $R$ is strictly increasing \cite[Proposition~3.4]{FoHu:16}. Following Proposition~3.7 in \cite{FoHu:16}, the $z$-derivatives are bounded by mostly polynomial multiples of themselves, i.e. $\abs{R_z} \leq \widetilde d_{01}(z) R$, $\abs{R_xz} \leq \widetilde d_{11}$ and $\abs{\vz_{xz}} \leq d_{11}(z)\vz_x$ with positive mostly polynomially growing $\widetilde d_{01}$, $\widetilde d_{11}$ and $d_{11}$. Thus, the above inequality is bounded by $d(z)x\vz_x$, which is then bounded by $d(z)\vz$ using the concavity of $\vz$.

		Then, one can use the Cauchy-Schwartz inequality to separate the functions depending only on $(y,z)$ from $\vz(s, X_s^\pz, Z_s)$ in the integral, i.e., it is reduced to
		\begin{equation}
		\left(\EE_{(t,y,z)} \int_t^T h^2(Y_s,Z_s)k^2(Z_s)\ud s \right)^{1/2}\left(\EE_{(t,x,y,z)}\int_t^T\vz(s,X_s^\pz,Z_s)^2\ud s \right)^{1/2}.
		\end{equation}
		Assumptions on $(Y_t, Z_t)$ ensure that the first part is uniformly bounded in $(\eps, \delta)$, while for the second part follows from Lemma~\ref{p2_lem_unibddadd}.
	 Similarly, the last two terms in \eqref{p2_eq_Eflprob} are bounded by repeating the above procedure using assumptions on the utility (cf. Assumption \ref{assump_U} equation \eqref{p1_assump_Uiii}). %Thus, we have shown that the auxiliary function $\widetilde E(t,x,y,z)$ is of order $\eps + \delta$.
		
		%	\todo{Additional assumptions are needed: equation \eqref{p1_assump_Uiii} is required for $2 \leq i \leq 7$. Consequently for the risk tolerance function, the Proposition \ref{p1_prop_ssh} equation \eqref{p1_eq_prop_Rbounds}  holds for $0\leq j \leq 6$, and $RR_{xxz}$, $R^2R_{xxxz}$, $R_{xzz}$, $RR_{xxzz}$ and $R^2R_{xxxzz}$ are bounded by constants.}
		
		So far, we have shown for any $(t,x,y,z) \in [0,T]\times \RR^+\times \RR \times \RR$
		\begin{equation}
		\abs{\widetilde E(t,x,y,z)} \leq \widetilde C\left(\eps + \delta + \sqrt{\eps\delta}\right) \leq \widetilde C(\delta + \eps),
		\end{equation}
		where $\widetilde C$ may varying from line to line and is free of $(\eps, \delta)$.  By the difference between $E$ and $\widetilde E$, one has
		\begin{multline}
		\abs{\Vfl - \vz - \sqrt{\eps}\voz - \sqrt{\delta}\vzo} \\  \leq \abs{\widetilde E} + \abs{\eps\vtzl + \eps^{3/2}\vthzl + \eps\sqrt{\delta}\vtol} \leq C(\eps+\delta),
		\end{multline}
		where $C = C(t,x,y,z)$ and is independent of small parameters $(\eps, \delta)$. Thus we obtain the desired result.

\section{The Asymptotic Optimality of $\pz$}\label{p6_sec_optimality}

This section focuses on the proof of Theorem~\ref{p6_thm_main}, dedicated to the performance of $\pz$ by comparison with other admissible strategy of the form
\begin{equation}
\pitilde = \pzt + \eps^\alpha \pozt + \delta^\beta\pzot.
\end{equation}
Detailed assumptions on any fixed choices of $\pzt, \pozt$ and $\pzot$ are given in Assumptions~\ref{p6_assump_piregularity} and \ref{p6_assump_optimality}.
Recall the corresponding value process $\Vft_t$ defined in \eqref{def_Vft}:
		\begin{equation}\label{p6_def_Vfpi}
		\Vft_t = \EE\{U(X_T^\pi)\vert \MCF_t\},
		\end{equation}
		with $\pi  = \pitilde $, and $X_t^\pi$ given by
		\begin{equation}\label{p6_def_Xtpi}
		\ud X_t = \pitilde_t\mu(Y_t, Z_t) \ud t + \pitilde_t \sigma(Y_t, Z_t) \ud W_t.
		\end{equation}
We would like to compare asymptotically the performance of $\pz$ with the one of  $ \pitilde $ by looking at the approximations of $\Vfl$ and $\Vft$. For $\Vfl$,  the rigorous result has been derived in Theorem~\ref{p2_Thm_full}. Thus, it remains to find approximations associated to $ \pitilde$ at a desired order.
Note that \eqref{p6_def_Vfpi} is a process rather than a function of $(t,x,y,z)$ as we do not restrict ourself to work with Markovian strategies. This is also explicitly stated in the following.
	\begin{assump}\label{p6_assump_piregularity}
		For a fixed choice of $(\pzt$, $\pozt$, $\pzot$, $\alpha$, $\beta)$, we require:
		\begin{enumerate}[(i)]
			\item The processes $\pzt_t$, $\pozt_t$ and $\pzot_t$ are adapted to the filtration $\MCF_t$ generated by the three Brownian motions $(W_t, W_t^Y, W_t^Z)$.
			\item The strategy $\pitilde = \pzt + \eps^\alpha\pozt +  \delta^\beta \pzot$ is admissible.
			\item The function $\mu(y,z)$ is at most polynomially growing.
			%\item  Functions $\pzt(t,x,z)$ and $\pot(t,x,z)$ are continuous on $[0,T]\times \RR^+\times \RR$, and $C^1$ in $z$.
			\item The process $\vz(t,X_t^\pi, Z_t) = M(t,X_t^\pi;\aves(Z_t))$ is in $L^4([0,T]\times \Omega)$ uniformly in $(\eps, \delta)$, i.e.,
			\begin{equation}
			\EE\left[\int_0^T \left(\vz(s,X_s^\pi,Z_s)\right)^4 \ud s\right] \leq C_2,
			\end{equation}
			where $C_2$ is independent of $(\eps,\delta)$, $Z_t$ follows \eqref{p2_eq_Ztfull}, and $X_t^\pi$ follows \eqref{p6_def_Xtpi}.
			%		\item Let  $(\widetilde{X}_s^{t,x})_{t\leq s\leq T}$ be the solution to:
			%		\begin{equation}\label{p3_eq_Xttilde}
			%		\ud \widetilde X_s = \mu(z)\pzt(s,\widetilde X_s, z) \ud s + \sigma(z) \pzt(s,\widetilde X_s,z) \ud W_s,
			%		\end{equation}
			%		starting at $x$ at time $t$.
			%		
			%		By (i), $\widetilde{X}_s^{t,x}$ is nonnegative and we further
			%		assume that it has full support $\RR^+$ for any $t<s\leq T$.
			%	
		\end{enumerate}
	\end{assump}

The above assumptions are mainly to ensure that $\Vft$ is well-defined, and heuristic expansions can be obtained. Once this is done, additional technical  integrability conditions on $\pitilde$ are needed, to rigorize the derivation. In order not to cut the presentation flow, we shall list them in Appendix~\ref{p6_appendix_addasump}.
	
	The first attempt of finding the approximation of $\Vft$ is to use the PDE approach, as in the case of $\Vfl$. In order to do so, we indeed need to restrict $\pitilde$ to feedback strategies, that is, $\pzt_t$, $\pzot_t$ and $\pzot_t$ are functions of $(t, X_t^\pi, Y_t, Z_t)$. Consequently $\Vft_t$ becomes a function of $(t,x,y,z)$, and can be characterized by a PDE. 
	Let $\MCL$ be the infinitesimal generator of the state processes $(X_t^\pi, Y_t, Z_t)$ with $X_t^\pi$ defined in \eqref{p6_def_Xtpi}, by definition $\Vft$ satisfies:
	\begin{equation}
	\partial_t \Vft + \MCL \Vft = 0, \quad \Vft(T,x,y,z) = U(x).
	\end{equation}
	According to the powers of $\eps$ and $\delta$, one can rewrite the generator $\MCL$ as:
	\begin{align}
	0 = \partial_t + \MCL &= \frac{1}{\eps}\MCL_0+ \delta \MCM_2 + \frac{\sqrt{\delta}}{\sqrt{\eps}}\MCM_3 + \frac{1}{\sqrt{\eps}}\widetilde \MCL_1 + \widetilde \MCL_2 +  \eps^\alpha\widetilde \MCL_3 + \eps^{2\alpha} \widetilde \MCL_4 + \eps^{\alpha-1/2}  \widetilde \MCL_5  \nonumber \\
	&\quad +\sqrt\delta\widetilde \MCM_1 + \delta^\beta\widetilde \MCM_4 + \delta^{2\beta}\widetilde \MCM_5 + \eps^\alpha\delta^\beta\widetilde \MCM_6 + \frac{\delta^\beta}{\sqrt{\eps}}\widetilde \MCM_7 + \eps^\alpha\sqrt\delta\widetilde \MCM_8 + \delta^{\beta + 1/2}\widetilde \MCM_9, \label{eq_Vft}
	\end{align}
	where the operators are defined as (arguments of $\pzt, \pozt, \pzot$ are systematically omitted for brevity):
	\begin{align}
	\widetilde \MCL_1 &= \rho_1a(y)\sigma(y,z)\pzt \partial_x\partial_y, &\widetilde \MCL_2 &= \partial_t + \mu(y,z)\pzt\partial_x +\half\sigma^2(y,z)\left(\pzt\right)^2\partial_x^2,&\\
	\widetilde \MCL_3 &= \mu(y,z)\pozt\partial_x +  \sigma^2(y,z)\pzt\pozt\partial_x^2, &\widetilde \MCL_4 &= \half\sigma^2(y,z)\left(\pozt\right)^2\partial_x^2,&\\
	\widetilde \MCL_5 &= \rho_1a(y)\sigma(y,z)\pozt\partial_x\partial_y, &\widetilde \MCM_1 &= \rho_2\sigma(y,z)g(z)\pzt\partial_x\partial_z,&\\
	\widetilde \MCM_4 &= \mu(y,z)\pzot\partial_x + \sigma^2(y,z)\pzt\pzot\partial_x^2, &\widetilde \MCM_5 &= \half\sigma^2(y,z)\left(\pzot\right)^2\partial_x^2,&\\
	\widetilde \MCM_6 &= \sigma^2(y,z)\pozt\pzot\partial_x^2, &\widetilde \MCM_7 &= \rho_1a(y)\sigma(y,z)\pzot\partial_x\partial_y,&\\
	\widetilde \MCM_8 &= \rho_2\sigma(y,z)g(z)\pozt\partial_x\partial_z, &\widetilde \MCM_9 &= \rho_2\sigma(y,z)g(z)\pzot\partial_x\partial_z.&
	\end{align}
	Observing that four scales $\eps^\alpha, \sqrt{\eps}, \delta^\beta, \sqrt{\delta}$ exist in \eqref{eq_Vft}, we propose the following ansatz
	\begin{equation}\label{p6_eq_Vftansatz}
	\Vft = \vzt + \sum_{i+j = 1}^{n+1} \eps^{i\alpha}\delta^{j\beta}\widetilde v^{(i\alpha,j\beta)} + \sqrt{\eps}\vozt + \sqrt\delta\vzot + \cdots,
	\end{equation}
	where $n = \max(n_1,n_2)$ and $n_1$ (\emph{resp.} $n_2$) is the largest integer satisfying $n_1\alpha < \half$ (\emph{resp.} $n_2\beta < \half$). If one follows the derivation in our previous work \cite[Section~4]{FoHu:16} where only the slow factor is considered, after having the ansatz, the next step is to identify the needed terms before $o (\sqrt{\eps} + \sqrt{\delta})$ in \eqref{p6_eq_Vftansatz} and justify the expansion. While doing this, keep in mind that we need to compare it to the approximation of $\Vfl$, which is $\vz + \sqrt{\eps}\voz + \sqrt{\delta}\vzo + \MCO(\eps + \delta)$. In some cases, the comparison is difficult. Even for the cases that the comparison can be done, this process is lengthy and tedious by matching terms of all different orders. For instance, when $\alpha = \frac{1}{8}$ and $\beta = \frac{3}{8}$, the terms clearly before $o(\sqrt{\eps} + \sqrt{\delta})$ in \eqref{p6_eq_Vftansatz} are
	\begin{equation}
	\vzt + \eps^\alpha \widetilde v^{1\alpha,0\beta} + \delta^\beta \widetilde v^{0\alpha,1\beta} + 
	\eps^{2\alpha}\widetilde v^{2\alpha,0\beta} + \eps^\alpha\delta^\beta \widetilde v^{\alpha, \beta}  + \eps^{3\alpha}\widetilde v^{3\alpha,0\beta} +\sqrt{\eps} \vozt +  \sqrt{\delta} \vzot,
	\end{equation}
	where $\widetilde v^{1\alpha,0\beta}$ and $\widetilde v^{0\alpha,1\beta}$ are identified zeros, and other terms can be characterized by effective equations. But it is impossible to compare every single terms above to $\vz + \sqrt{\eps}\voz + \sqrt{\delta}\vzo$, the first order approximation of $\Vfl$. However, if one add the term $\delta^{2\beta} \widetilde v^{0\alpha, 2\beta}$ to the above expression (although itself is $o(\sqrt{\eps} + \sqrt{\delta})$), and regroup it with $\eps^{2\alpha}\widetilde v^{2\alpha,0\beta} + \eps^\alpha\delta^\beta \widetilde v^{\alpha, \beta}$, one will be able to claim the sum is negative and of order $\MCO(\eps^{2\alpha} + \delta^{2\beta})$. Consequently, we can claim $\pz$ outperforms at the order $\eps^{2\alpha} + \delta^{2\beta}$, and there is no need to analyze further terms (e.g. $\eps^{3\alpha}\widetilde v^{3\alpha,0\beta} +\sqrt{\eps} \vozt +  \sqrt{\delta} \vzot$). As you can see, for just one case of $\alpha, \beta$, it is already quite tricky to do the comparison.
	As a result, in this section we shall present the optimality of $\pz$ by another approach: the epsilon-martingale decomposition method. One advantage of the epsilon-martingale decomposition method is the relaxation of the feedback form controls. As you have seen in the aforementioned assumptions on $\pitilde$, we do not require $\pitilde$ to be an explicit function of the states $(t, X_t, Y_t, Z_t)$, but rather a general adapted process, although we intend to compare it with $\pz$, which is of the Markovian type.

	\subsection{The Epsilon-Martingale Decomposition}\label{sec_EMD2_intro}
	The epsilon-martingale decomposition is an efficient tool to find approximations of martingales of interest in non-Markovian problems when small parameters are involved. Denote this martingale by $(V_t^\delta)_{t\in[0,T]}$ with respect to some filtration $(\MCF_t)_{t \in [0,T]}$, where $\delta$ represents the group of small parameters, then the method consists of making an ansatz $Q_t^\delta$ for $V_t^\delta$ in the form of a martingale plus something small (nonzero non-martingale part) with the right terminal condition. Then this ansatz is indeed the approximation to $V_t^\delta$ with an error that is of the order of the non-martingale part. More precisely, suppose one intends to find the approximation of $V_t^\delta$ at order $\sqrt\delta$ %(this will be the case in Section~\ref{p3_sec_optimality}), 
	then it requires to decompose $Q_t^\delta$ as:
	\begin{equation}\label{p3_eq_epsmart}
	Q_t^\delta = M_t^\delta + R_t^\delta, \quad \text{and} \quad Q_T^\delta = V_T^\delta,
	\end{equation}
	where $M_t^\delta$ is a martingale and $R_t^\delta$ is of order $o(\sqrt\delta)$. Note that the term of order $\sqrt\delta$ will be absorbed in the martingale part $M_t^\delta$. Suppose we obtain such a decomposition \eqref{p3_eq_epsmart}, then, taking conditional expectation with respect to $\MCF_t$ on both sides of the equation $Q_T^\delta = M_T^\delta + R_T^\delta$ gives
	\begin{equation}
	V_t^\delta = \EE\left[Q_T^{\delta} \vert \MCF_t\right] =
	M_t^\delta + \EE\left[R_T^\delta \vert \MCF_t \right] = Q_t^{\delta}  + \EE\left[R_T^\delta \vert \MCF_t \right] - R_t^\delta.
	\end{equation}
	Since $R_t^\delta$ is of order $o(\sqrt\delta)$, $Q_t^{\delta}$ is the approximation to $V_t^\delta$  up to $\sqrt\delta$. Therefore the above argument leads to the desired approximation result.
	%Now it remains to find $Q_t^{\pz, \delta}$ so that the decomposition holds, and we have the following proposition.
	
	In our case, the martingale considered is $\Vft_t$ defined in \eqref{p6_def_Vfpi}, and the desired order of $R^{\eps, \delta}$ is $o(\sqrt{\eps} + \sqrt{\delta})$ or $o(1)$ depending on the relation between $\pz$ and $\pzt$. The derivation will be presented in the next section.
	%In the following sections, this philosophy will be used repeatedly to analyze the portfolio performance of a given zeroth order strategy under various fractional stochastic environments, as well as to show its asymptotic optimality within a smaller class of admissible strategies. As a byproduct, when the results is applied to power utility case, we will obtain the optimality within the full class $\MCA$, which has been announced in advance in Section~\ref{sec_EMD}.

\subsection{Asymptotics of $\Vft$ and Proof of Theorem~\ref{p6_thm_main}}

		%In what follows, we shall show the derivations in details because it is different from the fractional stochastic environments cases. 
		
In what follows, we shall derive the approximation of $\Vft$. To shorten the length of expressions, we systematically omit the arguments $(t,X_t^\pi, Z_t)$ of the functions $\vz$, $\voz$ and $\vzo$ when no confused is introduced. Also, the claims on the true martingality and on the order of residual terms are guaranteed by Assumption~\ref{p6_appendix_addasump}. Since this is used through the derivation, we mentioned it at the beginning, and will not repeat this reasoning later on.

The order of approximation will depend on $\pzt$ being identical to $\pz$ or not. We first deal with the case $\pzt = \pz$. By the It\^{o}'s formula applied to $\vz(t,X_t^\pi, Z_t)$, we deduce
		\begin{align}
		\ud \vz(t, X_t^\pi, Z_t) &= \Ltx(\lambda(Y_t, Z_t))\vz \ud t + \left(\eps^\alpha \pozt_t + \delta^\beta \pzot_t\right) \mu(Y_t, Z_t) \vz_x \ud t \\
		& \quad +  \half \left(\eps^\alpha \pozt_t + \delta^\beta \pzot_t\right)^2 \sigma^2(Y_t, Z_t) \vz_{xx} \ud t \\
		& \quad + \lambda(Y_t, Z_t) R(t,X_t;\aves(Z_t))\left(\eps^\alpha \pozt_t + \delta^\beta \pzot_t\right) \sigma(Y_t, Z_t) \vz_{xx} \ud t  \\
		& \quad + \delta \MCM_2 \vz \ud t + \sqrt{\delta}\rho_2 g(Z_t)\pi_t \sigma(Y_t, Z_t) \vz_{xz} \ud t \\
		& \quad  + \ud \widetilde M_t^{(1)},
		\end{align}
		where $ \widetilde M_t^{(1)}$ is a martingale
		\begin{equation}
		\ud \widetilde M_t^{(1)} = \pi_t \sigma(Y_t, Z_t)\vz_x \ud W_t + \sqrt{\delta} g(Z_t) \vz_z\ud W_t^Z.
		\end{equation}
		Using the relations $\Ltx(\aves(z)))\vz = 0$, and the definitions of $R(t,x;\lambda)$ and $D_k$, one can simplify the above as
		\begin{align}
		\ud \vz(t, X_t^\pi, Z_t) &= \half \left(\lambda^2(Y_t, Z_t) - \aves^2(Z_t)\right) D_1 \vz \ud t + \sqrt{\delta}\rho_2g(Z_t)\lambda(Y_t, Z_t)D_1 \vz_{z} \ud t \\
		& \quad + \ud \widetilde N_t  + \ud \widetilde R^{(1)}_t + \ud \widetilde M_t^{(1)}, \label{p6_eq_dvz}
		\end{align}
		where $\widetilde N_t \sim \MCO(\eps^{2\alpha} + \delta^{2\beta})$ is strictly decreasing and $\widetilde R^{(1)}$ is higher than order $\sqrt{\eps} + \sqrt{\delta}$ defined by
		\begin{align}
		\ud \widetilde N_t &= \half \left(\eps^\alpha \pozt_t + \delta^\beta \pzot_t\right)^2 \sigma^2(Y_t, Z_t) \vz_{xx} \ud t,\\
		\ud \widetilde R_t^{(1)} &= \delta \MCM_2 \vz \ud t + \sqrt\delta \rho_2 g(Z_t)\left(\eps^\alpha\pozt_t + \delta^\beta\pzot_t\right) \sigma(Y_t, Z_t) \vz_{xz} \ud t.
		\end{align}
		
		Now it remains to find the epsilon-martingale decomposition for the first two terms in \eqref{p6_eq_dvz}. To this end, we first analyze the term $\left(\lambda^2(Y_t, Z_t) - \aves^2(Z_t)\right) \ud t$, which will be repeated used in the following derivation. Recall the solution $\theta(y,z)$ of the Poisson equation defined in Section~\ref{p2_sec_multiheuristic}: $\MCL_0 \theta(y,z) = \lambda^2(y,z) - \aves^2(z)$. Applying the It\^{o}'s formula to $\theta(Y_t, Z_t)$ and omit the arguments $(Y_t, Z_t)$ of $\theta$ for the sake of length, one deduces
		\begin{align}
		\ud \theta(Y_t, Z_t) &= \left[\frac{1}{\eps}\MCL_0\theta+ \delta \MCM_2 \theta + \sqrt{\frac{\delta}{\eps}}\MCM_3 \theta\right] \ud t  + \frac{1}{\sqrt{\eps}} a(Y_t) \partial_y\theta \ud W_t^Y + \sqrt{\delta}g(Z_t)\partial_z\theta \ud W_t^Z,
		\end{align}
		where we recall that $\MCL_0$ and $\MCM_2$ are infinitesimal generators of $Y^{(1)} \stackrel{\MCD}{=} Y_{t\eps}$ and $Z^{(1)}  \stackrel{\MCD}{=} Z_{t/\delta}$ and $\MCM_3 = \rho_{12}a(y)g(z)\partial_{y}\partial_z$. Therefore,
		\begin{align}
		\left(\lambda^2(Y_t, Z_t) - \aves^2(Z_t)\right)\ud t &= \eps \ud \theta - \left[\eps\delta \MCM_2 \theta + \sqrt{\eps\delta}\MCM_3\theta\right] \ud t - \sqrt{\eps}a(Y_t) \partial_y\theta \ud W_t^Y\\
		&\quad  - \eps\sqrt{\delta}g(Z_t)\partial_z\theta \ud W_t^Z.
		\end{align}
		Then the first term in \eqref{p6_eq_dvz} is computed as follows
		\begin{align}\label{p6_eq_dlambdad1vz}
		\half \left(\lambda^2(Y_t, Z_t) - \aves^2(Z_t)\right) D_1 \vz \ud t = \frac{\eps}{2} D_1\vz \ud \theta  - \half\ud \widetilde R_t^{(2)} - \half\ud \widetilde M_t^{(2)},
		\end{align}
		where $\widetilde R_t^{(2)}$ is again of order $o(\sqrt{\eps} + \sqrt{\delta})$ and $\widetilde M_t^{(2)}$ is a true martingale defined by
		\begin{align}
		\ud \widetilde R_t^{(2)} & = \left[\eps\delta \MCM_2 \theta + \sqrt{\eps\delta}\MCM_3\theta\right] D_1 \vz \ud t, \\
		\ud \widetilde M_t^{(2)}& =  \sqrt{\eps}a(Y_t) \theta_y D_1\vz \ud W_t^Y + \eps\sqrt{\delta}g(Z_t)\theta_z D_1 \vz \ud W_t^Z.
		\end{align}
		For the term $D_1\vz \ud \theta$, we use the product rule $\ud \left(D_1 \vz \theta\right) = D_1\vz \ud \theta + \theta \ud D_1 \vz + \ud \average{D_1\vz, \theta}_t$ and obtain
		\begin{align}\label{p6_eq_d1vzdtheta}
		\frac{\eps}{2} D_1\vz \ud \theta = -\frac{\sqrt{\eps}}{2} \rho_1 B(Z_t) D_1^2 \vz \ud t + \half\ud \widetilde R_t^{(3)} + \half \ud \widetilde M_t^{(3)},
		\end{align}
		with
		\begin{align}
		\ud \widetilde R_t^{(3)} & = \eps \ud \left(D_1\vz \theta\right) - \left(\sqrt\eps\rho_1\theta_y a(Y_t) + \eps\sqrt{\delta}\rho_2 \theta_z g(Z_t)\right)\left(\eps^\alpha\pozt_t + \delta^\beta\pzot_t\right)\\
		& \quad \times \sigma(Y_t, Z_t)\partial_xD_1\vz \ud t -\left(\sqrt{\eps\delta}\rho_{12}\theta_y a(Y_t) + \delta \theta_z g(Z_t)\right)\partial_z D_1\vz g(Z_t) \ud t  \\
		& \quad  - \eps \theta \left[\partial_t + \pi_t\mu(Y_t,Z_t)\partial_x + \half \pi_t^2\sigma^2(Y_t,Z_t)\partial_{x}^2 + \delta M_2  +\sqrt{\delta}\rho_2 g(Z_t)\pi_t\sigma(Y_t,Z_t)\partial_{xz}\right]D_1\vz \ud t  \\
		& \quad - \sqrt{\eps}\rho_1 \left(a(Y_t)\lambda(Y_t,Z_t) \theta_y-B(Z_t)\right)D_1^2 \vz \ud t, \\
		\ud \widetilde M_t^{(3)} & = \eps \pi_t \sigma(Y_t, Z_t) \partial_x D_1\vz \ud W_t + \eps\sqrt{\delta}g(Z_t)\partial_z D_1\vz \ud W_t^Z.
		\end{align}
		Now recall that the first order correction in the fast variable $\voz$ defined in \eqref{p2_eq_vozf} satisfies $\Ltx(\aves(z))\voz = \half \rho_1 B(z)D_1^2\vz$, therefore
		\begin{align}\label{p6_eq_dvoz}
		\ud \sqrt{\eps}\voz(t, X_t^\pi, Z_t) =  \frac{\sqrt{\eps}}{2} \rho_1 B(Z_t)D_1^2\vz \ud t  + \ud \widetilde R_t^{(4)} + \ud \widetilde M_t^{(4)},
		\end{align}
		where
		\begin{align}
		\ud \widetilde R_t^{(4)} &= \sqrt{\eps}\left(\eps^\alpha\pozt_t + \delta^\beta\pzot_t\right)\mu(Y_t, Z_t) \voz_x \ud t + \frac{\sqrt{\eps}}{2} \left(\eps^\alpha\pozt_t + \delta^\beta\pzot_t\right)^2\sigma^2(Y_t, Z_t) \voz_{xx}\ud t \\
		& \quad + \sqrt{\eps}\left(\eps^\alpha\pozt_t + \delta^\beta\pzot_t\right)\lambda(Y_t, Z_t)R(t,X_t^\pi, \aves(Z_t))\sigma(Y_t, Z_t)\voz_{xx}\ud t \\
		& \quad + \sqrt{\eps} \delta \MCM_2 \voz \ud t + \sqrt{\eps\delta} \rho_2 g(Z_t)\pi_t \sigma(Y_t,Z_t)\voz_{xz} \ud t \\
		& \quad + \frac{\sqrt{\eps}}{2}\left(\lambda^2(Y_t, Z_t) - \aves^2(Z_t)\right)(D_2 + 2D_1) \voz \ud t,\\
		\ud \widetilde M_t^{(4)} &= \sqrt{\eps}\pi_t \sigma(Y_t, Z_t) \voz_x \ud W_t + \sqrt{\eps\delta}g(Z_t)\voz_z \ud W_t^Z.
		\end{align}
		
		The second term in \eqref{p6_eq_dvz} is taken care of by the first order correction in the slow variable $\vzo$, which satisfies $\Ltx(\aves(z))\vzo = -\rho_2 \widehat\lambda(z)g(z)D_1\vz_z$; see \eqref{p2_eq_vzof}. Applying the It\^{o}'s formula to $\vzo(t, X_t^\pi, Z_t)$ yields
		\begin{align}\label{p6_eq_dvzo}
		\ud \sqrt{\delta}\vzo(t, X_t^\pi, Z_t) = \sqrt\delta\rho_2 g(Z_t)\lambda(Y_t, Z_t) D_1\vz_z \ud t + \ud \widetilde R_t^{(5)} + \ud \widetilde M_t^{(5)},
		\end{align}
		with
		\begin{align}
		\ud \widetilde R_t^{(5)} &= \sqrt{\delta}\left(\eps^\alpha\pozt_t + \delta^\beta\pzot_t\right)\mu(Y_t, Z_t) \vzo_x \ud t + \frac{\sqrt{\delta}}{2} \left(\eps^\alpha\pozt_t + \delta^\beta\pzot_t\right)^2\sigma^2(Y_t, Z_t) \vzo_{xx}\ud t \\
		& \quad + \sqrt{\delta}\left(\eps^\alpha\pozt_t + \delta^\beta\pzot_t\right)\lambda(Y_t, Z_t)R(t,X_t^\pi, \aves(Z_t))\sigma(Y_t, Z_t)\vzo_{xx}\ud t \\
		& \quad +  \delta^{3/2} \MCM_2 \vzo \ud t + \delta \rho_2 g(Z_t)\pi_t \sigma(Y_t,Z_t)\vzo_{xz} \ud t \\
		& \quad + \frac{\sqrt{\delta}}{2}\left(\lambda^2(Y_t, Z_t) - \aves^2(Z_t)\right)(D_2 + 2D_1) \vzo \ud t\\
		& \quad - \sqrt{\delta}\rho_2 g(Z_t)\left(\lambda(Y_t, Z_t) - \widehat \lambda(Z_t)\right)D_1\vz_z \ud t,\\
		\ud \widetilde M_t^{(5)} &= \sqrt{\delta}\pi_t \sigma(Y_t, Z_t) \vzo_x \ud W_t + \delta g(Z_t)\vzo_z \ud W_t^Z.
		\end{align}
	  	Now, define the function $Q(t,x,z)$ by
		\begin{equation}
		Q(t,x,z) = \vz(t,x,z) + \sqrt{\eps}\voz(t,x,z) + \sqrt{\delta}\vzo(t,x,z),
		\end{equation}
		whose terminal condition is $Q(T,x,z) = \vz(T,x,z) = U(x)$.
		Combing equation \eqref{p6_eq_dvz}, \eqref{p6_eq_dlambdad1vz}, \eqref{p6_eq_d1vzdtheta}, \eqref{p6_eq_dvoz} and \eqref{p6_eq_dvzo}, we deduce
		\begin{equation}
		\ud Q(t,X_t^\pi, Z_t) = \ud \widetilde R_t + \ud \widetilde M_t + \ud \widetilde N_t,
		\end{equation}
		where $\widetilde R_t$ is of order $o(\sqrt{\eps} + \sqrt{\delta})$, and $\widetilde M_t$ is a true martingale given by
		\begin{align}
		\ud \widetilde R_t &= \ud  \widetilde R_t^{(1)} -\half\ud  \widetilde R_t^{(2)} + \half\ud  \widetilde R_t^{(3)} +\ud  \widetilde R_t^{(4)} +\ud  \widetilde R_t^{(5)},\\
		\ud \widetilde M_t& = \ud \widetilde M_t^{(1)} -\half\ud \widetilde M_t^{(2)}+\half\ud \widetilde M_t^{(3)}+\ud \widetilde M_t^{(4)}+\ud \widetilde M_t^{(5)}.
		\end{align}
		The above claims on the order of $\widetilde R_t$ and on the true martingality of $\widetilde M_t$ are justified by integrability conditions required in Assumption~\ref{p6_assump_optimality}, estimates of $\vz$ listed in \cite[Proposition~3.7]{FoHu:16}, and growth conditions of various functions.
		Finally we conclude
		\begin{align}
		\Vft_t&= \EE[U(X_T^\pi)\vert \MCF_t] = \EE[Q(T, X_T^\pi, Z_T) \vert \MCF_t] \\
		&= Q(t, X_t, Z_t) + \EE[\widetilde R_T - \widetilde R_t \vert \MCF_t] + \EE[\widetilde N_T - \widetilde N_t \vert \MCF_t] \\
		& < Q(t, X_t, Z_t) + o(\sqrt{\eps} + \sqrt{\delta}),\label{p6_eq_Vfteq}
		\end{align}
		where the last step is by the monotonicity of $\widetilde N_t$.
		
		In the case that $\pzt \ne  \pz$, similar derivation brings
		\begin{align}
		\ud \vz(t, X_t^\pi, Z_t) &= \Ltx(\lambda(Y_t, Z_t))\vz \ud t + \half \left(\pi_t - \pz\right)^2 \sigma^2(Y_t, Z_t)\vz_{xx} \ud t + \delta \MCM_2 \vz \ud t \\
		& \quad + \sqrt{\delta}\rho_2g(Z_t)\pi_t \sigma(Y_t, Z_t) \vz_{xz} \ud t  + \pi_t\sigma(Y_t, Z_t) \vz_x \ud W_t + \sqrt{\delta} g(Z_t)\vz_z \ud W_t^Z \\
		& = \ud \widehat R_t + \ud \widehat N_t + \ud \widehat M_t,
		\end{align}
		where $\widehat R_t$ is of order $\MCO(\sqrt{\eps} + \sqrt{\delta})$, $\widehat M_t$ is a true martingale and $\widehat N_t$ is strictly decreasing and of order one:
		\begin{align}
		\ud \widehat R_t &= \half(\lambda^2(Y_t, Z_t) - \aves^2(Z_t)) D_1 \vz \ud t +\delta \MCM_2 \vz \ud t + \sqrt{\delta}\rho_2g(Z_t)\pi_t \sigma(Y_t, Z_t) \vz_{xz} \ud t, \\
		\ud \widehat N_t & = \half \left(\pi_t - \pz\right)^2 \sigma^2(Y_t, Z_t)\vz_{xx} \ud t, \\
		\ud \widehat M_t & = \pi_t\sigma(Y_t, Z_t) \vz_x \ud W_t + \sqrt{\delta} g(Z_t)\vz_z \ud W_t^Z.
		\end{align}
		Therefore in this case,
		\begin{align}
		\Vft_t&= \EE[U(X_T^\pi) \vert \MCF_t] = \EE[\vz(T, X_T^\pi, Z_T) \vert \MCF_t] \\
		&= \vz(t, X_t, Z_t) + \EE[\widehat R_T - \widehat R_t \vert \MCF_t] + \EE[\widehat N_T - \widehat N_t \vert \MCF_t] \\
		& < \vz(t, X_t, Z_t) + \MCO(\sqrt{\eps} + \sqrt{\delta}). \label{p6_eq_Vftneq}
		\end{align}
		
The inequality in Theorem~\ref{p6_thm_main} is a consequence of the two approximation results of $\Vfl$ and $\Vft$.  By comparing the approximation of $\Vfl$  given in Theorem~\ref{p2_Thm_full} with the definition of $Q(t,x,z)$, we deduce
		\begin{equation}
		\Vfl(t,x,y,z) = Q(t,x,z) + \MCO(\eps + \delta).
		\end{equation}
Now, compare it with the two inequalities \eqref{p6_eq_Vfteq} and \eqref{p6_eq_Vftneq}, and observing that $\frac{\widetilde N_T - \widetilde N_t }{\sqrt{\eps}  + \sqrt{\delta}}$ and $\frac{\widehat N_T - \widehat N_t }{\sqrt{\eps}  + \sqrt{\delta}}$ are negative no matter what values $\alpha$ and $\beta$ take, we have the desired result.

\section{Conclusion}\label{sec_conclusion}

In this paper, we study the portfolio optimization problem in multiscale stochastic environment when the investor's utility is general. Motivated by recent empirical studies \cite{FoPaSiSo:11}, the return and volatility of the underlying asset are modeled by  functions of both fast and slow time scales. We first analyze the performance of a zeroth order strategy proposed in \cite{FoSiZa:13}, and give a rigorous approximation of the value process associated to this strategy, up to the first order. Then we compare its performance to any admissible strategy of a specific form. The comparison is made up to a certain order, thus we call this result asymptotic optimality. The first part is done by applying the singular and regular perturbation techniques to a linear PDE; while the second part, we employ the epsilon-martingale decomposition method, which not only simplifies the derivation, but also extends the analysis to non-Markovian strategies. We comment that our results partially answer the question \eqref{def_object} by giving a suboptimal strategy via analyzing the associated linear PDE, although a full optimality result will require to work with viscosity solutions of the HJB equation. It is also of the authors' interest to extend the analysis to fractional multiscale environment, motivated by the recent studies \cite{roughvol}.

\appendix

%\section{Discussion on General Utility}\label{appendix_U}

\section*{Appendix}

\section{Additional Assumptions in Section~\ref{p6_sec_optimality}}\label{p6_appendix_addasump}
	
	This set of assumptions is used to establish the approximation accuracy \eqref{p6_eq_Vfteq} (\emph{resp.} \eqref{p6_eq_Vftneq}) to $\Vft$ defined in \eqref{p6_def_Vfpi}. To be specific, these assumptions will ensure that $\widetilde M_t$ (\emph{resp.} $\widehat M_t$) is a true martingale and that $\widetilde R_t$ (\emph{resp.} $\widehat R_t$)  is of order $o(\sqrt\eps + \sqrt{\delta})$ (\emph{resp.} $\MCO(\sqrt{\eps} + \sqrt{\delta})$).
	\begin{assump}\label{p6_assump_optimality}
		Let $\pitilde = \pzt+\eps^\alpha\pozt + \delta^\beta\pzot$ be the trading strategy to compare with, and recall that $X^\pi$ is the wealth process generated by this strategy $\pi = \pitilde$ as defined in \eqref{p6_def_Xtpi}. In order to condense the notation,  we systematically omit  the arguments $(s,X_s^\pi,Z_t)$ of $\vz$, $\voz$ and $\vzo$ and the argument $(Y_t, Z_t)$ of $\mu$ and $\sigma$ in what follows. According to the different cases, we further require:
		\begin{enumerate}[(i)]
			\item\label{p6_assump_optimality_eq} If $\pzt \equiv \pz$, the quantities below, for any $t \in [0,T]$,  %are of order $\eps^{1-H}$ in $L^1$ sense:	
			%$\int_t^T \phi_s^\eps \mu \pot_t (\partial_x + R(s,X_s^\pi;\overline{\lambda})\partial_{xx})D_1\vz \ud s$,
			%$\int_t^T \phi_s^\eps \sigma^2 (\pot_t)^2\partial_{xx}D_1\vz\ud s$,
			%and the following quantities
			are uniformly bounded in $(\eps, \delta)$:
			
			$\EE\int_0^T \left(\sigma\pozt_s\vz_x\right)^2 \ud s $, $\EE\int_0^T \left(\sigma\pzot_s\vz_x\right)^2 \ud s $, $\EE\int_0^T \left(\sigma^2(\pozt_s)^2\vz_{xx}\right)^2 \ud s $,
			
			$\EE\int_0^T \left(\sigma^2(\pzot_s)^2\vz_{xx}\right)^2 \ud s $,
			$\EE \abs{\int_0^T  \mu\pozt_s \voz_x \ud s}$,	$\EE \abs{\int_0^T  \mu\pzot_s \voz_x \ud s}$,	
			
			$\EE \abs{\int_0^T  \sigma^2\left(\pozt_s\right)^2 \voz_{xx} \ud s}$,	$\EE \abs{\int_0^T  \sigma^2\left(\pzot_s\right)^2 \voz_{xx} \ud s}$,	
			
			$\EE \abs{\int_0^T  \mu\pozt_s R(s, X_s^\pi; \overline{\lambda}(Z_s)) \voz_{xx} \ud s}$,	
			$\EE \abs{\int_0^T  \mu\pzot_s R(s, X_s^\pi; \overline{\lambda}(Z_s)) \voz_{xx} \ud s}$,	
			
			$\EE \int_0^T  \left(\sigma\pozt_s \voz_{x}\right)^2 \ud s$, $\EE \int_0^T  \left(\sigma\pzot_s \voz_{x}\right)^2 \ud s$, $\EE \abs{\int_0^T  \mu\pozt_s \vzo_x \ud s}$,	
			
			$\EE \abs{\int_0^T  \mu\pzot_s \vzo_x \ud s}$,	$\EE \abs{\int_0^T  \sigma^2\left(\pozt_s\right)^2 \vzo_{xx} \ud s}$,	$\EE \abs{\int_0^T  \sigma^2\left(\pzot_s\right)^2 \vzo_{xx} \ud s}$,	
			
			$\EE \abs{\int_0^T  \mu\pozt_s R(s, X_s^\pi; \overline{\lambda}(Z_s)) \vzo_{xx} \ud s}$,	
			$\EE \abs{\int_0^T  \mu\pzot_s R(s, X_s^\pi; \overline{\lambda}(Z_s)) \vzo_{xx} \ud s}$,	
			
			$\EE \int_0^T  \left(\sigma\pozt_s \vzo_{x}\right)^2 \ud s$, $\EE \int_0^T  \left(\sigma\pzot_s \vzo_{x}\right)^2 \ud s$,

			\item\label{p6_assump_optimality_neq} If $\pzt \not\equiv \pz$, we require 	$\EE\int_0^T \left(\sigma\pi_s\vz_x\right)^2 \ud s$ to be uniformly bounded in $\eps$ and $\delta$.
		\end{enumerate}
	\end{assump}

\bibliographystyle{plain}
\bibliography{Reference}

\end{document}